\documentclass{llncs}

\usepackage[utf8]{inputenc}
\usepackage{amssymb,amsfonts,amsmath}
\usepackage[all,arc]{xy}
\usepackage{mathrsfs}
\usepackage{comment}
\usepackage{graphicx} 
\usepackage{commath}
\usepackage{afterpage}
\usepackage{xcolor}
\usepackage{listings}
\makeatletter
\def\lst@visiblespace{$\;$}
\makeatother
\usepackage{paralist}
\usepackage{caption}
\usepackage{subcaption}
\usepackage{todonotes}

\usepackage[T1]{fontenc}

\usepackage[ruled, linesnumbered,vlined]{algorithm2e}

\newcommand{\expv}{\mathbb{E}}

\lstdefinelanguage{affprob}
{
	morekeywords={angel,demon, choice, prob, if, then, else, fi,
		while, do, od,
		true, false, and, or, skip, sample, assume},
	sensitive = false,
	escapeinside={@:}{:@}
}

\usepackage{tikz}
\usepackage{tikz,pgffor}
\usepackage{changepage}
\usetikzlibrary{arrows}
\usetikzlibrary{shapes}
\usetikzlibrary{calc}
\usetikzlibrary{automata}
\usetikzlibrary{positioning}
\usetikzlibrary{angles}

\tikzstyle{ang}=[regular polygon, regular polygon sides = 3,draw,inner sep=0pt,minimum size=6mm, yshift = -0.75 mm]
\tikzstyle{dem}=[shape=diamond,draw,inner sep=0pt,minimum size=6mm]
\tikzstyle{ran}=[shape=circle,draw,inner sep=0pt,minimum size=6mm]
\tikzstyle{det}=[shape=rectangle,draw,inner sep=0pt,minimum size=5mm]
\tikzstyle{tran}=[draw,->,>=stealth, rounded corners]
\tikzstyle{pt}=[shape=circle,draw,inner sep=0pt,minimum size=0mm]
\tikzstyle{trannoarrow}=[draw,>=stealth, rounded corners]

\usepackage{todonotes}
\renewcommand{\vec}[1]{\mathbf{#1}}

\newcommand{\term}{Term}
\newcommand{\Term}{\term}
\newcommand{\Termset}{\mathit{Terminates}}

\newcommand{\support}{\mathit{supp}}

\newcommand{\vars}{\mathcal{V}}

\newcommand{\up}{u}
\newcommand{\pCFG}{\mathcal{C}}

\newcommand{\cfg}[2]{\vec{C}^{#1}_{#2}}

\newcommand{\run}{\varrho}
\newcommand{\locinit}{\loc_{\mathit{init}}}
\newcommand{\vecinit}{\vec{x}_{\mathit{init}}}
\newcommand{\natfilt}{\mathcal{R}}
\newcommand{\lem}{\eta}

\newcommand{\stime}{T}

\newcommand{\ttime}{\mathit{Term}}

\newcommand{\updates}{\mathit{Up}}
\newcommand{\Fpath}{\mathit{Fpath}}

\newcommand{\Run}{\mathit{Run}}
\newcommand{\genpath}{\pi}

\newcommand{\minlev}{\mathit{min}\text{-}\mathit{lev}}

\newcommand{\indicator}[1]{{\mathbb{I}}_{#1}}
\newcommand{\locterm}{\loc_\lout}

\newcommand{\OmegaRun}{\Omega_{\mathit{Run}}}

\newcommand{\pvars}{V}

\newcommand{\locs}{\mathit{L}}

\newcommand{\loc}{\ell}

\newcommand{\tran}{\tau}

\newcommand{\E}{\mathbb{E}}

\newcommand{\Rset}{\mathbb{R}}

\newcommand{\lout}{\mathit{out}}

\newcommand{\probm}{\mathbb{P}}

\newcommand{\transitions}{\mapsto}

\newcommand{\guards}{G}

\newcommand{\conf}{c}

\newcommand{\eps}{\varepsilon}
\newcommand{\gentransitions}{\Delta}
\newcommand{\succDist}{\delta}

\newcommand{\Unif}{\mathit{Uni}}

\newcommand{\predicate}{\Psi}

\renewcommand{\epsilon}{\eps}
\renewcommand{\vars}{\pvars}

\newcommand{\rf}{f}
\newcommand{\state}{s}
\newcommand{\tr}{\tau}

\newcommand{\unneg}{\snneg}
\newcommand{\prank}{\textit{P-RANK}}

\newcommand{\unif}{\Unif}
\newcommand{\snneg}{\textit{S-NNEG}}
\newcommand{\pnneg}{\textit{P-NNEG}}

\newcommand{\expneg}{\textit{EXP-NNEG}}
\newcommand{\weakexpneg}{\textit{W-EXP-NNEG}}
\newcommand{\unbounded}{\textit{UNBOUND}}

\newcommand{\normaldist}{\mathit{Norm}}
\newcommand{\cutset}{S}

\begin{document}
	
	\title{On Lexicographic Proof Rules for Probabilistic Termination}
	
	\author{Krishnendu Chatterjee\inst{1} \and Ehsan Kafshdar Goharshady\inst{2}  \and Petr Novotn\'y \inst{3} \and Ji\v{r}\'i Z\'arev\'ucky \inst{3} \and \DJ or\dj e \v{Z}ikeli\'c\inst{1}}
	
	\institute{IST Austria, Klosterneuburg, Austria \\ \email{krishnendu.chatterjee@ist.ac.at, djordje.zikelic@ist.ac.at}\and Ferdowsi University of Mashhad, Mashhad, Iran \\ \email{e.kafshdargoharshady@mail.um.ac.ir} \and Masaryk University, Brno, Czech Republic \\ \email{petr.novotny@fi.muni.cz, xzarevuc@fi.muni.cz}}



\maketitle

\begin{abstract}
	We consider the almost-sure (a.s.) termination problem for probabilistic programs, which are a stochastic extension of classical imperative programs. Lexicographic ranking functions provide a sound and practical approach for termination of non-probabilistic programs, and their extension to probabilistic programs is achieved via lexicographic ranking supermartingales (LexRSMs). However, LexRSMs introduced in the previous work have a limitation that impedes their automation: all of their components have to be non-negative in all reachable states. This might result in LexRSM not existing even for simple terminating programs. Our contributions are twofold: First, we introduce a generalization of LexRSMs which allows for some components to be negative. This standard feature of non-probabilistic termination proofs was hitherto not known to be sound in the probabilistic setting, as the soundness proof requires a careful analysis of the underlying stochastic process. Second, we present polynomial-time algorithms using our generalized LexRSMs for proving a.s. termination in broad classes of linear-arithmetic programs.
	
\end{abstract}


\section{Introduction}
\label{sec:intro}

The extension of classical imperative programs with randomization gives rise to probabilistic programs (PPs)~\cite{GHNR14:prob-programming}, which are used
in multitude of applications, including 
stochastic network protocols~\cite{BaierBook,prism,netkat,netkat2},
randomized algorithms~\cite{RandBook,RandBook2}, 
security~\cite{BGGHS16:diff-privacy-coupling,BGHP16:diff-privacy-siglog} ,
machine learning, and planning~\cite{LearningSurvey,G15,roy2008stochastic,gordon2013model,scibior2015practical,claret2013bayesian,thrun2002probabilistic}.
The analysis of PPs is an active research area in formal methods \cite{SriramCAV,pmaf,pldi18,AgrawalC018,CFNH16:prob-termination,CFG16,EsparzaGK12,KKMO16:wp-expected-runtime,OKKM16:recursive-prob-wp-calculus,kaminski2018hardness}.
PPs can be extended with nondeterminism to allow 
over-approximating program parts that are too complex for static analysis~\cite{DBLP:conf/popl/CousotC77,MM05}.


For non-probabilistic programs, the \emph{termination} problem asks whether
a given program {\em always} terminates. While the problem is well-known to be undecidable over Turing-complete programs, many sound automated techniques that work well for practical programs have been developed~\cite{CPR06:terminator-journal,CPR11:termination-cacm}. Such techniques typically seek a suitable \emph{certificate} of termination.
Particularly relevant certificates are \emph{ranking functions (RFs)}~\cite{rwfloyd1967programs,DBLP:conf/cav/BradleyMS05,DBLP:conf/tacas/ColonS01,DBLP:conf/vmcai/PodelskiR04,DBLP:conf/pods/SohnG91,PR04:transition-invariants} mapping program states into a well-founded domain, forcing a strict decrease of the function value in every step. The basic ranking functions are 1-dimensional, which is 
often insufficient for complex control-flow structures.
Lexicographic ranking functions (LexRFs) are multi-dimensional extensions 
of RFs that provide an effective approach to 
termination analysis~\cite{CSZ13,ADFG10:lexicographic,GMR15:rank-extremal,BCIKP16:T2,DBLP:conf/cav/BradleyMS05,BCF13}.  The literature typically restricts to linear LexRFs for linear-arithmetic (LA) programs, as LA reasoning can be more efficiently automated compared to non-linear arithmetic.

For probabilistic programs, the termination problem considers 
aspects of the probabilistic behaviors as well. The most fundamental 
is the \emph{almost-sure (a.s.)} termination problem, which asks whether a given PP terminates with probability~1.
One way of proving a.s.\ termination is via \emph{ranking supermartingales (RSMs)},
a probabilistic analogue of ranking functions named so due to the connection with (super)martingale stochastic  processes~\cite{Williams:book}.
There is a rich body of work on 1-dimensional RSMs, while the work~\cite{AgrawalC018} introduces lexicographic RSMs. 
In probabilistic programs, a transition $ \tr $ available in some state $ \state $
yields a probability distribution over the successor states. The conditions defining RSMs are formulated in terms of the expectation operator $ \expv^\tr $ of this distribution. In particular,  \emph{lexicographic ranking supermartingales} (LexRSMs) of \cite{AgrawalC018} are functions $ f $ mapping program states to $ \Rset^d $, such that for each transition $ \tran $ there exists a component $ 1 \leq i \leq d $, satisfying, for any reachable state $s$ at which $\tau$ is enabled, the following conditions $ \prank $ and $ \unneg $ (with $ \rf_i $ the $ i $-component of $ \rf $ and $ \state \models\guards(\tr) $ denoting the fact that $ \state $ satisfies the guard of $ \tr $):
\begin{enumerate}
\item $ \prank(f,\tran) \equiv 
\state \models\guards(\tr) \Rightarrow \big(\expv^\tran[f_i(\state')] \leq f_i(\state) -1 $ and $ \expv^\tran[f_j(\state')] \leq f_j(\state)$ for all $ 1 \leq j < i \big)$.
\item $ \unneg(f,\tau) \equiv \state \models\guards(\tr) \Rightarrow  \big(f_j(\state)\geq 0$ for all $ 1 \leq j \leq d \big)$. 
\end{enumerate}
(We use the standard primed notation from program analysis, i.e.\ $ \state' $ is the probabilistically chosen successor of $ \state $ when performing $ \tr $.)
The $ \prank $ condition enforces an expected decrease in lexicographic ordering, while $ \snneg $ stands for ``strong non-negativity''.
Proving the soundness of LexRSMs for proving a.s.\ termination is highly non-trivial and requires reasoning about complex stochastic processes~\cite{AgrawalC018}.
Apart from the soundness proof,~\cite{AgrawalC018} also presents an algorithm for the synthesis of linear LexRSMs.

While LexRSMs improved the applicability of a.s.\ termination proving, their usage is impeded by the \emph{restrictiveness of strong non-negativity}
due to which a linear LexRSM might not exist even for simple a.s.\ terminating programs. This is a serious drawback from the automation perspective, since even if such a program admits a non-linear LexRSM, efficient automated tools that restrict to linear-arithmetic reasoning would not be able to find it. 


\lstset{language=affprob}
\lstset{tabsize=5,showspaces}

\newsavebox{\intromotcut}
\begin{lrbox}{\intromotcut}

\end{lrbox}

\newsavebox{\intromot}
\begin{lrbox}{\intromot}

\end{lrbox}

\begin{figure}[t]
	\centering
\begin{subfigure}{0.49\textwidth}
\centering
\begin{lstlisting}[mathescape]
$ \ell_0 $: while $ y\geq 0 $ do
$ \phantom{\loc_0:}$     $ x:=y; $
$ \ell_1 $:	 while $ x\geq 0 $ do
$ \phantom{\loc_0:}$       $ x:=x-1+\normaldist(0,1) $
$ \phantom{\loc_0:} $     od;
$ \phantom{\loc_0:}$     $ y:=y-1 $
$ \phantom{\loc_0:} $ od
\end{lstlisting}
\caption{}
\label{fig:intromot:a}
\end{subfigure}
\hfill
\begin{subfigure}{0.49\textwidth}
\centering
\begin{lstlisting}[mathescape]
$ \ell_0 $: while $ x\geq 0 $ do
$ \phantom{\loc_0:} $   if $ y\geq 0 $ then
$ \phantom{\loc_0:} $     $ y:=y+\unif[-7,1] $
$ \phantom{\loc_0:} $   else
$ \phantom{\loc_0:} $     $ x:=x+\unif[-7,1] $;
$ \ell_1 $:     $ y:= y  + \unif[-7,1] $
$ \phantom{\loc_0:} $   fi od
\end{lstlisting}
\caption{}
\label{fig:intromot:b}
\end{subfigure}
\caption{Motivating examples. 
$ \normaldist(\mu,\sigma) $ samples from the normal distribution with mean $ \mu $ and std.\ deviation $ \sigma $. 
$ \unif[a,b] $ samples uniformly from the interval $ [a,b] $. Location labels are the ``$ \loc_i $'': one location per loop head and one additional location in (b) so as to have one assignment per transition (a technical requirement for our approach). A formal representation of the programs via \emph{probabilistic control flow graphs} is presented later, in Section~\ref{sec:prog-prelim}.}
\label{fig:intromot}
\end{figure}

Consider the program in Figure~\ref{fig:intromot:a}. 
By employing simple random-walk arguments, we can manually prove that the 
program terminates a.s. A linear LexRSM proving this needs to have a component containing a positive multiple of $ x $ at the head of the inner while-loop ($ \loc_1 $).
However, due to the sampling from the normal distribution, which has unbounded support, the value of $ x $ inside the inner loop cannot be bounded from below. Hence, the program does not admit a linear LexRSM. In general, LexRSMs with strong non-negativity do not handle well programs with unbounded-support distributions.

Now consider the program in Figure~\ref{fig:intromot:b}. It can be again shown that this PP terminates a.s.; however, this cannot be witnessed by a linear LexRSM: to rank the ``if-branch'' transition, there must be a component with a positive multiple of $ y $ in $ \loc_0 $. But $ y $ can become arbitrarily negative within the else branch, and cannot be bounded from below by a linear function of $ x $.

\smallskip\noindent{\em Contribution: Generalized Lexicographic RSMs.}
In the non-probabilistic setting, strong non-negativity can be relaxed to \emph{partial non-negativity} ($ \pnneg $), where only the components which are to the left of the ``ranking component'' $ i $ (inclusive) need to be non-negative (Ben-Amram--Genaim RFs~\cite{BG15:lexicographic-complexity}). We show that in the probabilistic setting, the same relaxation is possible under additional \emph{expected leftward non-negativity} constraint $ \expneg $. Formally, we say that  $ \rf $ is a \emph{generalized lexicographic ranking supermartingale} (GLexRSM) if for any transition $ \tr $ there is $1 \leq  i \leq d $ such that for any reachable state $s$ at which $\tau$ is enabled we have $ \prank(\rf,\tr) \wedge \pnneg(\rf,\tr)\wedge \expneg(\rf,\tr) $, where
\begin{align*}
\pnneg(\rf,\tau) \quad&\equiv\quad\state \models\guards(\tr) \Rightarrow  \big( \rf_j(\state)\geq 0 \text{ for all } 1 \leq j \leq i \big)\\
\expneg(\rf,\tau) \quad&\equiv\quad\state \models\guards(\tr) \Rightarrow \big(\expv^\tr[\rf_j(s')\cdot \indicator{<j}(s')] \geq 0 \text{ for all } 1 \leq j \leq i \big),
\end{align*}
with $ \indicator{< j} $ being the indicator function of the set of all states in which a transition ranked by a component $ < j $ is enabled.

We first formulate GLexRSMs as an abstract proof rule for general stochastic processes. We then instantiate them into the setting of probabilistic
 programs and define \emph{GLexRSM maps,} which we prove to be sound for proving a.s.\ termination. These results are general and \emph{not specific} to linear-arithmetic programs. 
\smallskip\noindent{\em Contribution: Polynomial Algorithms for Linear GLexRSMs.} 
\begin{compactenum}
    \item For linear arithmetic PPs in which sampling instructions use bounded-sup\-port distributions we show that the problem \textsc{LinGLexPP} of deciding whether a given PP with a given set of \emph{linear invariants} admits a linear GLexRSM is decidable in polynomial time. 
    Also, our algorithm computes the witnessing linear GLexRSM whenever it exists. In particular, our approach proves the a.s. termination of the program in Fig.~\ref{fig:intromot:b}.
    \item Building on results of item 1, we construct a sound polynomial-time algorithm for a.s.~termination proving in PPs that \emph{do perform} sampling from \emph{un\-bounded-support} distributions. 
    In particular, the algorithm proves a.s.~termination for our motivating example in Fig.~\ref{fig:intromot:a}.
\end{compactenum}


\smallskip\noindent{\em Related work.} Martingale-based termination literature mostly focused on 1-di\-men\-sional RSMs~\cite{SriramCAV,CFNH16:prob-termination,CFG16,CNZ17,HolgerPOPL,MM16:proofrule-arxiv,MMKK18,HFC18,CF17,MBKK:21:amber-esop,GieslGH:19:constant-prob-programs}. RSMs themselves can be seen as generalizations of Lyapunov ranking functions from control theory~\cite{BG05,Foster53}. Recently, the work~\cite{Huang0CG19} pointed out the unsoundess of the 1-dimensional RSM-based proof rule in~\cite{HolgerPOPL} due to insufficient lower bound conditions and provided a corrected version.
On the multi-dimensional front, it was shown in~\cite{HolgerPOPL} that requiring components of (lexicographic) RSMs to be nonnegative only at points where they are used to rank some enabled transition (analogue of Bradley-Manna-Sipma LexRFs~\cite{DBLP:conf/cav/BradleyMS05}) is unsound for proving a.s.\ termination.
This illustrates the intricacies of dealing with lower bounds in the design of a.s.\ termination certificates.
Lexicographic RSMs with strong non-negativity were introduced in~\cite{AgrawalC018}. The work~\cite{ChenH20}
	produces an $\omega$-regular decomposition of program's control-flow graph, with each program component ranked by a different RSM. This approach does not require a lexicographic ordering of RSMs, but each component in the decomposition must be ranked by a single-dimensional non-negative RSM. 
RSM approaches were also used for cost analysis~\cite{pldi18,WFGCQS19,AvaMS:2020:mod-prob-cost} and additional liveness and safety properties~\cite{CVS16:martingale-recurrence-persistence,BEFH16:doob,CNZ17}.

Logical calculi for reasoning about properties of 
	probabilistic programs (including termination) were studied 
	in~\cite{Kozen:prob-semantics,FH:prdl,Kozen:probabilistic-PDL,Feldman:propositional-probdl}
	and extended to programs with non-determinism 
	in~\cite{MM04,MM05,KKMO16:wp-expected-runtime,OKKM16:recursive-prob-wp-calculus,GKI14:prob-semantics}. In particular~\cite{MM04,MM05,MMKK18} formalize RSM-like proof certificates within the \emph{weakest pre-expectation (WPE)} calculus~\cite{morgan1996probabilistic,morgan1999pgcl}.
	 The power of this calculus allows for reasoning about complex programs ~\cite[Section 5]{MMKK18}, but the proofs typically require a human input. Theoretical connections between martingales and the WPE calculus were recently explored in~\cite{HarkKGK20:aiming-low-journal}. There is also a rich body of work on analysis of probabilistic functional programs, where the aim is typically to obtain a general type system~\cite{LagoG17:prob-term-monadic-size,AvanziniLG19:type-based-prob-complexity,KobayashiLG19:prob-termination-higher-order,dLFR21:intersection-types-past} for reasoning about termination properties (automation for discrete probabilistic term rewrite systems was shown in~\cite{AvanzinidLY20:prob-term-rewriting}).

As for other approaches to a.s.\ termination, for \emph{finite-state programs} with nondeterminism 
a sound and complete method was given in~\cite{EsparzaGK12}, while~\cite{DBLP:conf/sas/Monniaux01} considers a.s.\ termination proving through abstract interpretation. The work~\cite{kaminski2018hardness} shows that proving a.s.\ termination is harder (in terms of arithmetical hierarchy) than proving termination of non-probabilistic programs.
%
%
%

	The computational complexity of the construction of 
	lexicographic ranking functions in non-probabilistic programs was studied 
	in~\cite{BG13:integer-ranking,BG15:lexicographic-complexity}.

\smallskip\noindent{\em Paper organization.} The paper is split in two parts: the first one is ``abstract'', with mathematical preliminaries~(Section~\ref{sec:prelim}) and definition and soundness proof of abstract GLexRSMs (Section~\ref{sec:glexrsm-pruning}). We also present an example showing that ``GLexRSMs'' without the expected leftward non-negativity constraint are not sound. The second part covers application to probabilistic programs: preliminaries on the program syntax and semantics~(Section~\ref{sec:prog-prelim}), a GLexRSM-based proof rule for a.s. termination (Section~\ref{sec:glexrsm-progs}), and the outline of our algorithms~(Section~\ref{sec:algo}).

\section{Mathematical Preliminaries}\label{sec:prelim}

\newsavebox{\exapp}
\begin{lrbox}{\exapp}
\begin{lstlisting}[mathescape]
while $x\geq 0$ and $y\geq 0$ do
	if $\mathbf{\star}$ then
		$x := x + \textbf{sample(}\Unif\{-3,1\}\textbf{)}$
	else 
		$x := \textbf{ndet}[0,\infty)$
		if $\textbf{prob(}0.5\textbf{)}$ then
			$ y:= y - 4$
		else
			$ y:= y + 2 $
		fi
	fi
od
\end{lstlisting}
\end{lrbox}

We use boldface notation for
vectors, e.g.\ $\vec{x}$, $\vec{y}$, etc., and we denote an $i$-th component of a
vector $\vec{x}$ by $\vec{x}[i]$. For an 
$n$-dimensional vector 
$\vec{x}$, index $1 \leq i\leq n$, and number $a$ we denote by $\vec{x}(i\leftarrow a)$ 
a 
vector $\vec{y}$ such that $\vec{y}[i]=a$ and $\vec{y}[j]=\vec{x}[j]$ for all 
$1\leq j \leq n$, $j\neq i$.
For two real numbers $a$ and $b$, we use $a \cdot b$ to denote their product.

	We assume familiarity with basics of probability theory~\cite{Williams:book}. A \emph{probability space} is a triple
	$(\Omega,\mathcal{F},\probm)$, where $\Omega$ is a
	\emph{sample space}, $\mathcal{F}$ is a \emph{sigma-algebra} of measurable
	sets over $\Omega$, and
	$\probm$ is a \emph{probability measure} on $\mathcal{F}$. A \emph{random variable (r.v.)} $ R:\Omega\rightarrow \mathbb{R}\cup\{\pm\infty\} $ is an $ \mathcal{F} $-\emph{measurable} real-valued function (i.e. $  \{\omega\mid R(\omega)\leq x\} \in \mathcal{F}$ for all $ x\in \Rset $) and we denote by $\expv[R]$ its \emph{expected value}.
	A \emph{random vector} is a vector whose every component is a random
	variable. We denote by $\vec{X}[j]$ the $j$-component of a random vector $\vec{X}$. A (discrete time) \emph{stochastic process} in a
	probability space $(\Omega,\mathcal{F},\probm)$ is an infinite sequence of
	random vectors in this space.
	We will also use random variables of the form $R\colon\Omega \rightarrow A$ for some finite or countable set $A$, which easily translates to the real-valued variables.

Let $(\Omega,\mathcal{F},\mathbb{P})$ be a probability space and let $X$ be a random variable. A \emph{conditional expectation} of $ X $ given a sub-sigma algebra $ \mathcal{F}' \subseteq \mathcal{F} $ is any real-valued random variable $Y$ s.t.: i) $ Y $ is $\mathcal{F}'$-measurable; and ii) for each set $A\in 
\mathcal{F}'$ it holds that $ \E[X\cdot \mathbb{I}(A)] = \E[Y\cdot \mathbb{I}(A)] $. Here, $\mathbb{I}(A) \colon \Omega\rightarrow \{0,1\}$ is an \emph{indicator function} of 
$A$, i.e. function returning $1$ for 
each $\omega\in A$ and $0$ for each $\omega\in \Omega\setminus A$.

It is known~\cite{Ash:book} that a random variable satisfying the properties of conditional expectation exists whenever a) $\E[|X|]<\infty$, i.e.~$X$ is {\em integrable}, or b) $X$ is real-valued and nonnegative (though these two conditions are not necessary). Moreover, whenever the conditional expectation exists it is also known to be a.s.\ unique. We denote this a.s.\ unique conditional expectation by $ \E[X|\mathcal{F}'] $. It holds that for any $ \mathcal{F}' $-measurable r.v.\ $ Z $ we have $\E[X\cdot Z|\mathcal{F}'] = \E[X|\mathcal{F}']\cdot Z$, whenever the former conditional expectation exists~\cite[Theorem 9.7(j)]{Williams:book}.

A \emph{filtration} in $(\Omega,\mathcal{F},\mathbb{P})$ is an increasing (w.r.t.\ set inclusion) sequence  $\{\mathcal{F}_t \}_{t=0}^{\infty} $ of sub-sigma-algebras of $\mathcal{F}$. A \emph{stopping time} w.r.t.\ a filtration $\{\mathcal{F}_t \}_{t=0}^{\infty} $ is a random variable $\stime$ taking values in $\mathbb{N}\cup \{\infty\}$ s.t.\ for every $t$ the set $\{\stime = t\} =  \{\omega \in \Omega \mid \stime(\omega) = t \}$ belongs to $\mathcal{F}_t$. Intuitively, $\stime$ returns a time step in which some process should be ``stopped'', and the decision to stop is made solely on the information available at the current step. 


\section{Generalized Lexicographic Ranking Supermartingales}\label{sec:glexrsm-pruning}

In this section, we introduce {\em generalized lexicographic ranking supermartingales (GLexRSMs):} an abstract concept that is not necessarily connected to PPs, but which is crucial for the soundness of our new proof rule for a.s. termination.

\begin{definition}[Generalized Lexicographic Ranking Supermartingale]\label{def:genlexrsm}
	Let $(\Omega,\mathcal{F},\mathbb{P})$ be a probability space and let $(\mathcal{F}_t)_{t=0}^{\infty}$ be a filtration of $\mathcal{F}$. Suppose that $T$ is a stopping time w.r.t.\ $ \mathcal{F} $. An $n$-dimensional real valued stochastic process $(\mathbf{X}_t)_{t=0}^{\infty}$ is a {\em generalized lexicographic ranking supermartingale for $T$} (GLexRSM) if:
	\begin{enumerate}
		\item For each $t\in\mathbb{N}_0$ and $1\leq j\leq n$, the random variable $\mathbf{X}_t[j]$ is $\mathcal{F}_t$-measurable.
		\item For each $t\in\mathbb{N}_0$, $1\leq j\leq n$, and $A\in\mathcal{F}_{t+1}$, the conditional expectation $\mathbb{E}[\mathbf{X}_{t+1}[j]\cdot \mathbb{I}(A)\mid \mathcal{F}_t]$ exists.
		\item For each $t\in\mathbb{N}_0$, there exists a partition of the set $\{T>t\}$ into $n$ subsets $L^t_1,\dots,L^t_n$, all of them $\mathcal{F}_t$-measurable (i.e., belonging to $ \mathcal{F}_t $), such that for each $1\leq j\leq n$
		\begin{itemize}
			\item $\mathbb{E}[\mathbf{X}_{t+1}[j] \mid \mathcal{F}_t](\omega) \leq \mathbf{X}_t[j](\omega)$ for each $\omega\in \cup_{j'=j}^n L^t_{j'}$,
			\item $\mathbb{E}[\mathbf{X}_{t+1}[j] \mid \mathcal{F}_t](\omega) \leq \mathbf{X}_t[j](\omega) - 1$ for each $\omega\in L^t_j$,
			\item $\mathbf{X}_t[j](\omega)\geq 0$ for each $\omega \in\cup_{j'=j}^n L^t_{j'}$,
			\item $\mathbb{E}[\mathbf{X}_{t+1}[j]\cdot \mathbb{I}(\cup_{j'=0}^{j-1}L^{t+1}_{j'}) \mid \mathcal{F}_i](\omega)\geq 0$ for each $\omega \in \cup_{j'=j}^n L^t_{j'}$, with $L^{t+1}_0=\{T\leq t+1\}$.
		\end{itemize}
	\end{enumerate}
\end{definition}

Intuitively, we may think of each $\omega\in\Omega$ as a trajectory of process that evolves over time (in the second part of our paper, this will be a probabilistic program run). Then, $\mathbf{X}_t$ is a vector function depending on the first $t$ time steps (each $\mathbf{X}_t[j]$ is $\mathcal{F}_t$-measurable), while $T$ is the time at which the trajectory is stopped. Then in point 3 of the definition, the first two items encode the expected (conditional) lexicographic decrease of $\mathbf{X}_t$, the third item encodes non-negativity of components to the left (inclusive) of the one which ``ranks'' $ \omega $ in step $ t $, and the last item encodes the expected leftward non-negativity (sketched in Section~1). For each $1\leq j\leq n$ and time step $t\geq 0$, the set $L^t_j$ contains all $\omega\in \{T> t\}$ which are ``ranked'' by the component $j$ at time $ t $.
 An \emph{instance} of an $n$-dimensional GLexRSM $\{\vec{X}_t \}_{t=0}^{\infty}$ is a tuple $(\vec{X}_{t=0}^{\infty}, \{L_1^t,\dots,L_n^t\}_{t=0}^{\infty})$,
where the second component is a sequence of partitions of $\Omega$ satisfying the condition in Definition~\ref{def:genlexrsm}. We say that $\omega\in \Omega$ has {\em level} $j$ in step $t$ of the instance $((\mathbf{X}_t)_{t=0}^{\infty},(L^t_1,\dots,L^t_n)_{t=0}^{\infty})$ if $T(\omega)>t$ and $\omega\in L^t_j$. If $T(\omega)\leq t$, we say that the level of $\omega$ at step $t$ is $0$.

We now state the main theorem of this section, which underlies the soundness of our new method for proving almost-sure termination.

\begin{theorem}\label{thm:genlexrsm}
	Let $(\Omega,\mathcal{F},\mathbb{P})$ be a probability space, $(\mathcal{F}_t)_{t=0}^{\infty}$ a filtration of $\mathcal{F}$ and $T$ a stopping time w.r.t.\ $ \mathcal{F} $. If there is an instance $((\mathbf{X}_t)_{t=0}^{\infty},(L^t_1,\dots,L^t_{n})_{t=0}^{\infty})$ of a GLexRSM over $(\Omega,\mathcal{F},\mathbb{P})$ for $T$, then $\mathbb{P}[T<\infty]=1$.
\end{theorem}
    

In~\cite{AgrawalC018}, a mathematical notion of LexRSMs is defined and a result for LexRSMs analogous to our Theorem~\ref{thm:genlexrsm} is established. Thus, the first part of our proof mostly resembles the proof of Theorem 3.3.\ in~\cite{AgrawalC018}, up to the point of defining the stochastic process $(Y_t)_{t=0}^{\infty}$ in eq.~\eqref{eq:processy}. After that, the proof of~\cite{AgrawalC018} crucially relies on nonnegativity of each $\mathbf{X}_t[j]$ and $Y_t$ at every $\omega\in \Omega$ that is guaranteed by LexRSMs, and it cannot be adapted to the case of GLexRSMs. Below we first show that, for GLexRSMs, $\mathbb{E}[Y_t]\geq 0$ for each $t\geq 0$, and then we present a very elegant argument via the Borel-Cantelli lemma~\cite[Theorem 2.7]{Williams:book} which shows that this boundedness of expectation is sufficient for the theorem claim to hold.


\begin{proof}[Sketch of proof of Theorem~\ref{thm:genlexrsm}] We proceed by contradiction. Suppose that there exists an instance of a GLexRSM but that $\mathbb{P}[T=\infty]>0$. First, we claim that there exists $1\leq k\leq n$ and $s,M\in\mathbb{N}_0$ such that the set $B$ of all $\omega\in \Omega$ for which the following properties hold has positive measure, i.e.~$\mathbb{P}[B]>0$: (1)~$T(\omega)=\infty$, (2)~$\mathbf{X}_s[k](\omega)\leq M$, (3)~for each $t\geq s$, the level of $\omega$ at step $t$ is at least $k$, and (4) the level of $\omega$ equals $k$ infinitely many times. The claim is proved by several applications of the union bound, see Appendix~\ref{app:genlglexrsm}.

Since $B$ is defined in terms of tail properties of $\omega$ (``level is at least $k$ {\em infinitely many times}'') it is not necessarily $\mathcal{F}_t$-measurable for any $t$. Hence, we define a stochastic process $(Y_t)_{t=0}^{\infty}$ such that each $Y_t$ is $\mathcal{F}_t$-measurable, and which satisfies the desirable properties of $(\mathbf{X}_t[k])_{t=0}^{\infty}$ on $B$.

Let $D=\{\omega\in \Omega \mid \mathbf{X}_s[k](\omega)\leq M\land\omega\in \cup_{j=k}^n L^s_j \}$. Note that $D$ is $\mathcal{F}_t$-measurable for $t\geq s$. We define a stopping time $F$ w.r.t.~$(\mathcal{F}_t)_{t=0}^{\infty}$ via $F(\omega) = \inf\{t\geq s \mid \omega\not\in\cup_{j'=k}^n L^t_{j'}\}$; then a stochastic process $(Y_t)_{t=0}^{\infty}$ via
\begin{equation}\label{eq:processy}
Y_t(\omega) = \begin{cases}
0, &\mbox{if } \omega\not\in D,\\
M, &\mbox{if } \omega\in D,\, \text{ and }t<s,\\
\mathbf{X}_t[k](\omega), &\mbox{if } \omega\in D, \text{ } t\geq s \text{ and } F(\omega) > t,\\
\mathbf{X}_{F(\omega)}[k](\omega), &\mbox{else}.
\end{cases}
\end{equation}
A straightforward argument (presented in Appendix~\ref{app:genlglexrsm}) shows that for each $t\geq s$ we have $\mathbb{E}[Y_{t+1}] \leq \mathbb{E}[Y_t] - \mathbb{P}[L^t_k\cap D\cap\{F>t\}]$. By a simple induction we obtain:
\begin{equation}\label{eq:ind}
\mathbb{E}[Y_s] \geq \mathbb{E}[Y_t] + \sum_{r=s}^{t-1}\mathbb{P}[L^{r}_k\cap D\cap\{F>r\}].
\end{equation}
Now, we show that $\mathbb{E}[Y_t]\geq 0$ for each $t\in\mathbb{N}_0$. The claim is clearly true for $t<s$, so suppose that $t\geq s$. We can then expand $\mathbb{E}[Y_t]$ as follows
\begin{align*}
&\mathbb{E}[Y_t] = \mathbb{E}[Y_t\cdot \mathbb{I}(F=s)] + \sum_{r=s+1}^t\mathbb{E}[Y_t\cdot \mathbb{I}(F=r)] + \mathbb{E}[Y_t\cdot \mathbb{I}(F>t)]\\
&\hspace{0.5cm} \text{($Y_s\geq 0$ as $D\subseteq \cup_{j=k}^n L^s_j$ and $Y_t(\omega)\geq 0$ whenever $F(\omega)>t$)}\\
&\geq \sum_{r=s+1}^t\mathbb{E}[Y_t\cdot \mathbb{I}(F=r)] = \sum_{r=s+1}^t\mathbb{E}[Y_t\cdot \mathbb{I}(\{F=r\}\cap D)]\\
&\hspace{0.5cm} \text{($Y_t(\omega) = \mathbf{X}_{F(\omega)}[k](\omega)$ whenever $\omega\in D$, $t\geq s$ and $F(\omega)\leq t$)}
\end{align*}
\begin{align*}
&= \sum_{r=s+1}^t\mathbb{E}[\mathbf{X}_r[k]\cdot \mathbb{I}(\cup_{j=0}^{k-1}L^r_j)\cdot \mathbb{I}(\{F>r-1\}\cap D)] \\
&\hspace{0.5cm} \text{(properties of cond. exp. \& $\mathbb{I}(\{F>r-1\}\cap D)$ is $\mathcal{F}_{r-1}$-measurable)}\\
&= \sum_{r=s+1}^t\mathbb{E}\Big[\mathbb{E}[\mathbf{X}_r[k]\cdot \mathbb{I}(\cup_{j=0}^{k-1}L^r_j)\mid\mathcal{F}_{r-1}]\cdot \mathbb{I}(\{F>r-1\}\cap D)\Big] \geq 0\\
&\hspace{0.5cm} \text{($\mathbb{E}[\mathbf{X}_r[k]\cdot \mathbb{I}(\cup_{j=0}^{k-1}L^r_j)\mid\mathcal{F}_{r-1}](\omega)\geq 0$ for $\omega\in \{F>r-1\}\subseteq \cup_{j=k}^n L^{r-1}_j$)}.
\end{align*}
Plugging into eq.~\eqref{eq:ind} that $\mathbb{E}[Y_t]\geq 0$, we get $\mathbb{E}[Y_s] \geq \sum_{r=s}^{t-1}\mathbb{P}[L^r_k\cap D\cap\{F>r\}]$ for each $t\geq s$. By letting $t\rightarrow\infty$, we conclude
$\mathbb{E}[Y_s] \geq \sum_{r=s}^{\infty}\mathbb{P}[L^r_k\cap D\cap\{F>r\}].$
As $Y_s\leq M$ and $Y_s=0$ outside $D$, we know that $\mathbb{E}[Y_s]\leq M\cdot\mathbb{P}[D]$. We get
\begin{equation*}
\sum_{r=s}^{\infty}\mathbb{P}[L^r_k\cap D\cap\{F=\infty\}] \leq \sum_{r=s}^{\infty}\mathbb{P}[L^r_k\cap D\cap\{F>r\}] \leq M\cdot\mathbb{P}[D] < \infty.
\end{equation*}
By the Borel-Cantelli lemma, $\mathbb{P}[L^r_k \cap D\cap \{F=\infty\}\text{ for infinitely many $r$}] = 0$.
But the event $\{L^r_k \cap D\cap \{F=\infty\}\text{ for infinitely many } r\}$ is precisely the set of all runs $\omega\in\Omega$ for which (1)~$T(\omega)=\infty$ (as $ \omega $ never has level zero by $ \omega \in L^r_k $ for inf.\ many $ k $), (2)~$\mathbf{X}_s[k](\omega)\leq M$, (3)~for each $r\geq s$ the level of $\omega$ at step $t$ is at least $k$, and (4) the level of $\omega$ is $k$ infinitely many times. Hence, $B=\{L^r_k \cap D\cap \{F=\infty\}\text{ for infinitely many } r\}$ and $\mathbb{P}[B]=0$, a contradiction.\qed
\end{proof}

GLexRSMs would be unsound without the expected leftward nonnegativity.

\begin{example}
Consider a one-dimensional stochastic process $ (Y_t)_{t=0}^{\infty} $ s.t.\ $ Y_0 = 1 $ with probability 1 and then the process evolves as follows: in every step $ t $, if $ Y_t\geq 0 $, then with probability $p_t = \frac{1}{4}\cdot\frac{1}{2^t} $ we put $ Y_{t+1} = Y_t - \frac{2}{p_t} $ and with probability $ 1 - p_t $ we put $ Y_{t+1} = Y_t + \frac{1}{1 - p_t} $. If $ Y_t < 0 $, we put $ Y_{t+1} = Y_{t} $. The underlying probability space can be constructed by standard techniques and we consider the filtration $ (\mathcal{F}_t)_{t=0}^{\infty} $ s.t. $ \mathcal{F}_t $ is the smallest sub-sigma-algebra making $ Y_t $ measurable. Finally, consider the stopping time $ \stime $ returning the first point in time when $ Y_t<0 $. Then $ \stime < \infty $ if and only if the process ever performs the update $ Y_{t+1} = Y_t - \frac{2}{p_t} $, but the probability that this happens is bounded by $ \frac{1}{4} +\frac{3}{4}\cdot\frac{1}{8}+\frac{3}{4}\cdot \frac{7}{8}\cdot\frac{1}{16}+\cdots < \frac{1}{4}\sum_{t=0}^\infty \frac{1}{2^t} = \frac{1}{2} <1 $. At the same time, putting $ L^t_1 = \{Y_t\geq 0\} $ we get that the tuple $ ((Y_t)_{t=0}^{\infty},(L^t_1)_{t=0}^{\infty}) $ satisfies all conditions of Definition~\ref{def:genlexrsm} apart from the last bullet of point 3. 
\end{example}

\section{Program-Specific Preliminaries}\label{sec:prog-prelim}

Arithmetic \emph{expressions} in our programs are 
built from constants, program variables and 
standard Borel-measurable~\cite{Billingsley:book} arithmetic operators. We also allow sampling instructions to appear on right-hand sides of variable assignments as linear terms. 
An expression with no such terms is called \emph{sampling-free}. We allow sampling from both discrete and continuous distributions. We denote by $\mathcal{D}$ the set of distributions appearing in the program with each $d\in\mathcal{D}$ assumed to be {\em integrable}, i.e.~$\mathbb{E}_{X\sim d}[|X|]<\infty$. This is to ensure that expected value of each $ d $ over any measurable set is well-defined and finite.


A \emph{predicate} over a set of variables $ V $ is a Boolean combination of 
\emph{atomic predicates} of the form $E\leq E'$, where $E$, $E'$ are 
sampling-free expressions whose all variables are from $ V $. We denote by $\vec{x}\models \predicate$ the fact that the predicate $\predicate$ is satisfied by 
substituting values of $\vec{x}$ for the corresponding variables in $\predicate$.

\begin{figure}[t]
	\centering
	\begin{tikzpicture}
	\node[ran] (l0) at (0,0)  {$\loc_0$};
	\node[ran, below = 1.6cm of l0] (l1) {$\loc_1$};
	\node[ran, left = 1.5cm of l0] (term) {$ \locterm $};
	\draw[tran] (l0) to node[font=\scriptsize,draw, fill=white, 
	rectangle,pos=0.5] {$y<0$} (term);
	\draw[tran] (l0) to node[font=\scriptsize,draw, fill=white, 
	rectangle,pos=0.4] {$y\geq 0$} node[auto, font = \scriptsize, right, pos=0.7] {$(x,u_1)$} (l1);
	\node[right = 1cm of l1, circle, minimum size = 3mm] (dum) {};
	\node[above=1.5cm of dum.east, coordinate] (dumx) {};
	\draw[tran, rounded corners] (l1) -- (dum.east|-l1) --  node[font=\scriptsize,draw, fill=white, 
	rectangle,pos=0.3] {$x< 0$} node[auto, font = \scriptsize, right, pos=0.7] {$(y,\up_3)$} (dumx) -- (l0.-45);
	\node[above left = 1.5cm of l1, circle, minimum size = 3mm] (dum1) {};
	\node[below left = 1.5cm of l1, circle, minimum size = 3mm] (dum2) {};
	\draw[tran, rounded corners] (l1) -- (dum1.west|-l1) -- node[font=\scriptsize,draw, fill=white, 
	rectangle,pos=0.3] {$x\geq 0$} node[auto, font = \scriptsize, left, pos=0.7] {$(x,\up_2)$} (dum2.west|-dum1) -- (l1);
	\begin{scope}[xshift=6cm]
	\node[ran] (l0) at (0,0)  {$\loc_0$};
	\node[below = 1.6cm of l0, circle, minimum size = 3mm] (dum1) {};
	\node[ran, left = 1.6cm of dum1] (l1) {$\loc_1$};
	\node[right = 1.6cm of dum1, circle, minimum size = 3mm] (dum2) {};
	\node[ran, left = 1.5cm of l0] (term) {$ \locterm $};
	\draw[tran] (l0) to node[font=\scriptsize,draw, fill=white, 
	rectangle,pos=0.5] {$x<0$} (term);
	\draw[tran] (l0) -- node[font=\scriptsize,draw, fill=white, 
	rectangle,pos=0.5] {$x\geq 0 \land y\geq 0$}  (dum2.east) -- node[auto, font = \scriptsize, above, pos=0.5] {$(y,u_1)$} (dum1.east) -- (l0);
	\draw[tran] (l1) -- node[auto, font = \scriptsize, above, pos=0.5] {$(y,u_3)$} (dum1.west) -- (l0.-100);
	\draw[tran] (l0) to node[font=\scriptsize,draw, fill=white, 
	rectangle,pos=0.55] {$x\geq 0 \land y< 0$} node[ font = \scriptsize, left, pos=0.3] {$(x,u_2)$} (l1);
	\end{scope}
	\end{tikzpicture}
	\caption{The pCFGs of the programs presented in Figure~\ref{fig:intromot}. Guards are shown in the rounded boxes, (absence of a  box = guard is $ \mathit{true} $). The update tuples are shown using variable aliases instead of indexes for better readability. On the left, we have $ \up_1 = y, \up_2 = x-1+\normaldist(0,1)$, and $ \up_3 = y - 1  $. On the right, we have $ \up_1 = y+\Unif[-7,1], \up_2 = x+\Unif[-7,1]$, and $ \up_3 = y +\Unif[-7,1]  $}
	\label{fig:running}
\end{figure}
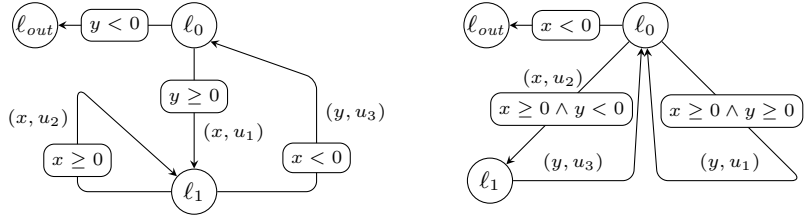


\label{sec:semantics}

We represent probabilistic programs (PPs) via the standard concept of \emph{probabilistic control flow graphs (pCFGs)}~\cite{AgrawalC018,CNZ17,CFNH16:prob-termination}. Formally, a (pCFG) is a tuple $\pCFG=(\locs,\pvars,\gentransitions,\updates,\guards)$ 
	where
	$\locs$ is a finite set of \emph{locations};
	$\pvars=\{x_1,\dots,x_{|\pvars|}\}$ is a finite set of \emph{program 
	variables};
	$\Delta$ is a finite set of \emph{transitions}, 
	i.e.~tuples of the form $\tau = (\loc,\succDist)$, where $\loc$ is a location and $\succDist$ is a distribution over \emph{successor locations}. $\Delta$ is partitioned into two disjoint sets: $\Delta_{PB}$ of probabilistic branching transitions for which $|\support(\succDist)|=2$, and $\Delta_{NPB}$ of remaining transitions for which $|\support(\succDist)|=1$.
%
%
	Next, 
	$\updates$ is a function assigning to each transition in $\Delta_{NPB}$ either the element $\bot$ (representing no variable update) or a  tuple $(i,\up)$, where $1\leq i \leq |\pvars|$ is a 
	\emph{target 
		variable index} and $\up$ 
	is an 
	\emph{update element}, which can 
	be either 
	an expression (possibly involving a single sampling instruction), or
	a bounded interval $R\subseteq \Rset$ representing a nondeterministic update.
	Finally, $\guards$ is a function assigning a predicate
	(a \emph{guard}) over $\pvars$ to each transition in $\Delta_{NPB}$. 
	Figure~\ref{fig:running} presents the pCFGs of our two motivating examples in Figure~\ref{fig:intromot}.
	
	Transitions in $\Delta_{PB}$ correspond to the ``probabilistic branching'' specified by the $\text{\textbf{if} }
\text{\textbf{prob($ p $)} } \text{\textbf{then}} \ldots$ $ \text{\textbf{else}}\ldots$ construct in imperative-style source code~\cite{AgrawalC018}. 
A program (pCFG) is \emph{linear} (or \emph{affine}) if all its expressions are \emph{linear,} i.e.\ of the form $ b + \sum_{i=1}^n a_i\cdot {Z_i}$ for constants $ a_1,\ldots,a_n,b $ and program variables/sampling instructions $Z_i$. we assume that parameters of distributions are constants, so they do not depend on program variable values, a common assumption in  martingale-based automated approaches to a.s. termination proving~\cite{SriramCAV,CFNH16:prob-termination,AgrawalC018,Huang0CG19,ChenH20}.

	
	A \emph{state} of a pCFG $\pCFG$ is a tuple $(\ell,\vec{x})$,
	where $\ell$ is a location of $\pCFG$ and $\vec{x}$ is a 
	$|\pvars|$-dimensional vector of \emph{variable valuations}. A transition $\tau$ is \emph{enabled} in $(\ell,\vec{x})$ if $\tau$ is outgoing from $\ell$ and $\vec{x}\models \guards(\tau)$. A state $\conf'=(\loc',\vec{x}')$ is a \emph{successor} of a 
	state $\conf=(\loc,\vec{x})$ if it can result from $\conf$ by performing a transition $\tau$ enabled in $\conf$ (see Appendix~\ref{app:succstates} for a formal definition).
	
	A \emph{finite path} of length $k$ in $\pCFG$ is a finite sequence $(\ell_0,\vec{x}_0)\cdots(\ell_k,\vec{x}_k)$  of
	states such that 
	$\loc_0=\locinit$ and 
	for each
	$0 \leq i < k$ the state $(\loc_{i+1},\vec{x}_{i+1})$ is a successor of 
	$(\loc_i,\vec{x}_i)$.
	A \emph{run} in
	$\pCFG$ is an infinite sequence of states whose every finite
	prefix is a finite path. We denote by $\Fpath_\pCFG$ and $\Run_\pCFG$ the 
	sets of all finite paths and runs in $\pCFG$, respectively. A state $ (\loc,\vec{x}) $ is \emph{reachable}  if there is, for some $ \vecinit $, a finite path starting in $ (\locinit,\vecinit) $ and ending in $ (\loc,\vec{x}). $ 
	
	
	The nondeterminism is resolved via schedulers. A \emph{scheduler} is a function $\sigma$
	assigning: i) to every finite path ending in a state 
		$\state$, a probability distribution over transitions enabled in $ \state $; and
		ii) to every finite path that ends in a state in which a transition $\tau$ with a nondeterministic update $ \updates(\tau) = (i,R) $ is enabled, an integrable probability distribution over $ R $.
	To make the program dynamics under a given scheduler well-defined, we restrict to \emph{measurable} schedulers. This is standard in probabilistic settings~\cite{DBLP:conf/fossacs/NeuhausserSK09,NK:CTMDP-bisimulation} and hence we omit the formal definition.
	

	
	
	We use the standard  Markov Decision Process (MDP) semantics of pCFGs \cite{KKMO16:wp-expected-runtime,CFNH16:prob-termination,AgrawalC018}. Each pCFG $\pCFG$ induces a sample space $\Omega_{\pCFG} = \Run_\pCFG$ and the standard \emph{Borel} sigma-algebra $\mathcal{F}_{\pCFG}$ over $\Omega_{\pCFG}$. Moreover, a pCFG $\pCFG$ together
		with a
		scheduler $\sigma$, initial location $\locinit$,
		and initial variable valuation
		$\vecinit$ uniquely determine a
				probability measure $\probm^{\sigma}_{\locinit,\vecinit}$ in the probability space  $(\Omega_\pCFG,\mathcal{F}_{\pCFG},\probm^{\sigma}_{\locinit,\vecinit})$ capturing the rather intuitive dynamics of the programs execution: we start in state $ (\locinit,\vecinit) $ and in each step, a transition $ \tr $ enabled in the current state is selected (using $ \sigma $ if multiple transitions are enabled). If $ \updates(\tr) = (i,u) $, then the value of variable $ x_i $ is changed according to $ u $. 
		The formal construction
		of $\probm^{\sigma}_{\locinit,\vecinit}$ proceeds via the standard 
		\emph{cylinder
			construction}~\cite[Theorem 2.7.2]{Ash:book}. We denote by
		$\E^\sigma_{\locinit,\vecinit}$ the expectation operator in the probability
		space
		$(\Omega_\pCFG,\mathcal{F}_{\pCFG},\probm^{\sigma}_{\locinit,\vecinit})$. 

We stipulate that each pCFG has a special \emph{terminal location} $ \locterm $ whose all outgoing transitions must be self-loops. We say that a run $\run$ \emph{terminates} if it contains a configuration whose first component is $\locterm$. We denote by $\Termset$ the set of all terminating runs in $\Omega_\pCFG$. We say that a program represented by a pCFG $\pCFG$ terminates \emph{almost-surely (a.s.)} if for each measurable scheduler $\sigma$ and each initial variable valuation $\vecinit$ it holds that $\probm^{\sigma}_{\locinit,\vecinit}[\Termset] = 1$. 

\section{GLexRSMs for Probabilistic Programs}\label{sec:glexrsm-progs}


In this section, we define a syntactic proof rule for a.s.\ termination of PPs, showing its soundness via  Theorem~\ref{thm:genlexrsm}. In what follows, let $\pCFG$ be a pCFG.

\begin{definition}[Measurable map]
	An $n$-dimensional measurable map (MM) is a vector 
	$\boldsymbol{\lem}=(\lem_1,\dots,\lem_n)$, where each $ \lem_i$ is a function mapping each location $\loc$ to a real-valued Borel-measurable function $\lem_i(\loc)$ over program variables. 
%
	We say that $ \boldsymbol{\lem} $ is a \emph{linear expression map} (LEM) if each $ \lem_i $ is  representable by a linear expression over program variables.
\end{definition}

The notion of pre-expectation was introduced in~\cite{Kozen:probabilistic-PDL}, was made syntactic in the Dijkstra wp-style in~\cite{morgan1999pgcl}, and was extended to programs with continuous distributions in~\cite{SriramCAV}. It formalizes the ``one-step'' expectation operator $ \expv^\tr $ we used on an intuitive level in the introduction.
%
%
In Appendix~\ref{app:preexpectation}, we generalize the definition of pre-expectation presented in~\cite{SriramCAV} in order to allow taking expectation over subsets of successor states $\pCFG$ (a necessity for handling the $ \expneg $ constraint). We say that a set $S$ of states in $\pCFG$ is {\em measurable}, if for each location $\loc$ in $\pCFG$ we have that $\{\mathbf{x}\in\mathbb{R}^{|\vars|}\mid (\loc,\mathbf{x})\in S\}\in \mathcal{B}(\mathbb{R}^{|\vars|})$, i.e.\ it is in the Borel sigma-algebra of $\mathbb{R}^{|\vars|}$. Furthermore, we also differentiate between the {\em maximal} and {\em minimal pre-expectation}, which may differ in the case of non-deterministic assignments in programs and intuitively are equal to the maximal resp.~minimal value of the next-step expectation over all non-deterministic choices. Let $\lem$ be a 1-dimensional MM, $ \tau = (\loc, \succDist) $ a transition and $S$ be a measurable set of states in $\pCFG$.
We denote by $ \text{max-pre}_{\lem,S}^\tau(s) $ the \emph{maximal pre-expectation} of $ \lem $ in $ \tau $ given $S$ (i.e. the maximal expected value of $ \lem $ after making a step from $ s $ computed over successor states belonging to $ S $), and similarly we denote by $ \text{min-pre}_{\lem,S}^\tau $ the \emph{minimal pre-expectation} of $ \lem $ in $ \tau $ given $S$.

As in the case of non-probabilistic programs, termination certificates are supported by program invariants over-approximating the set of reachable states. An {\em invariant} in $\pCFG$ is a function $I$ which to each location $\loc$ of $\pCFG$ assigns a Borel-measurable set $I(\loc)\subseteq \mathbb{R}^{|\vars|}$ such that for any state $(\loc,\mathbf{x})$ reachable in $\pCFG$ it holds that $\mathbf{x}\in I(\loc)$. If each $I(\loc)$ is given by a conjunction of linear inequalities over program variables, we say that $I$ is a {\em linear invariant}.



\paragraph{GLexRSM-Based Proof Rule for Almost-Sure Termination.}

Given $n\in \mathbb{N}$, we call a map $\mathsf{lev}:\gentransitions\rightarrow \{0,1,\dots,n\}$ a {\em level map}. For  $\tau\in\gentransitions$ we say that $\mathsf{lev}(\tau)$ is its level. The level of a state is the largest level of any transition enabled at that state. We denote by $S_{\mathsf{lev}}^{\leq j}$ the set of states  with level $\leq j$.



\begin{definition}[GLexRSM Map]\label{def:genlexrsm-map}
	Let $\boldsymbol{\lem}$ be an $n$-dimensional MM and $I$ an invariant in $\pCFG$.
	We say that $ \boldsymbol{\lem} $ is a \emph{generalized lexicographic ranking supermartingale map (GLexRSM map)} supported by $I$, if there is a level map $\mathsf{lev}:\gentransitions\rightarrow \{0,1,\dots,n\}$ such that $\mathsf{lev}(\tau)=0$ iff $\tau$ is a self-loop transition at $\locterm$, and for any transition $\tau=(\loc,\delta)$ with $\loc\neq \locterm$ the following conditions hold:
	\begin{compactenum}
		\item $ \prank(\boldsymbol{\lem},\tran) \equiv \mathbf{x}\in I(\loc)\cap \guards(\tau) \Rightarrow \big(\text{max-pre}_{\lem_{\mathsf{lev}(\tau)}}^\tau(\loc,\vec{x}) \leq \lem_{\mathsf{lev}(\tau)}(\loc,\mathbf{x}) -1 \wedge 
		 \text{max-pre}_{\lem_{j}}^\tau(\loc,\vec{x}) \leq \lem_{j}(\loc,\mathbf{x})$ for all $ 1 \leq j < \mathsf{lev}(\tau) \big)$; 
		\item $\pnneg(\boldsymbol{\lem},\tau) \equiv \mathbf{x}\in I(\loc)\cap\guards(\tau) \Rightarrow  \big(\lem_{j}(\loc,\mathbf{x})\geq 0 \text{ for all } 1 \leq j \leq \mathsf{lev}(\tau)\big)$;
		\item $\expneg(\boldsymbol{\lem},\tau) \equiv \mathbf{x}\in I(\loc)\cap\guards(\tau) \Rightarrow \text{min-pre}_{\lem_j,S^{\leq j-1}_{\mathsf{lev}}}^\tau(\loc,\vec{x}) \geq 0$ for all $ 1 \leq j \leq \mathsf{lev}(\tau)$.
	\end{compactenum}
	A GLexRSM map $ \boldsymbol{\lem} $ is \emph{linear} (or LinGLexRSM map) if it is also an LEM.
\end{definition}

\begin{theorem}[Soundness of GLexRSM-maps for a.s.~termination]\label{thm:genlexrsmmap}
	Let $\pCFG$ be a pCFG and $I$ an invariant in $\pCFG$. Suppose that $\pCFG$ admits an $n$-dimensional GLexRSM map $\boldsymbol{\lem}$ supported by $I$, for some $n\in\mathbb{N}$. Then $\pCFG$ terminates a.s.
\end{theorem}

The previous theorem, proved in Appendix~\ref{app:map-soundness}, instantiates Theorem~\ref{thm:genlexrsm} to probability spaces of pCFGs. The instantiation is \emph{not} straightforward.
To ensure that a scheduler cannot ``escape'' ranking by intricate probabilistic mixing of transitions, we prove that it is sufficient to consider \emph{deterministic} schedulers, which do not randomization among transitions. Also, previous martingale-based certificates of a.s.~termination~\cite{HolgerPOPL,CFNH16:prob-termination,CF17,AgrawalC018} often impose either nonnegativity or integrability of random variables defined by measurable maps in programs to ensure that their conditional expectations exist. We show that these conditional expectations exist even without such assumptions and in the presence of nondeterminism.
This generalizes the result of~\cite{SriramCAV} to PPs with nondeterminism.

\begin{remark}[Comparison to~\cite{Huang0CG19}]
	The work~\cite{Huang0CG19} considers a modular approach. Given a loop whose body has already been proved a.s.\ terminating, they show that the loop terminates a.s.\ if it admits a 1-dimensional MM satisfying $\prank$ for each transition in the loop, $\pnneg$ for the transition entering the loop, and the ``{\em bounded expected difference}'' property for all transitions. Hence, their approach is suited mainly for programs with incremental variable updates.
	
	Modularity is also a feature of the approaches based on the weakest pre-expectation calculus~\cite{MM04,MM05,MMKK18}.
%
\end{remark}

\section{Algorithm for Linear Probabilistic Programs}\label{sec:algo}


We now present two algorithms for proving a.s.\ termination in linear probabilistic programs (LinPPs). The first algorithm considers LinPPs with sampling from bounded-support distributions, and we show that the problem of deciding the existence of LinGLexRSM maps for such LinPPs is decidable. Our second algorithm extends the first algorithm into a sound a.s. termination prover for general LinPPs. In what follows, let $\pCFG$ be a LinPP and $I$ a linear invariant in $\pCFG$.

\subsection{Linear Programs with Distributions of Bounded Support}\label{sec:boundeddis}

Restricting to linear arithmetic is standard in automated a.s.~termination proving, allowing to encode the existence of the termination certificate into systems of linear constraints~\cite{SriramCAV,CFNH16:prob-termination,AgrawalC018,ChenH20}. In the case of LinGLexRSM maps, the difficulty lies in encoding the $\expneg$ condition, as it involves integrating distributions in variable updates which cannot always be done analytically. We show, however, that for LinPPs with bounded-support sampling, we can define another condition which is easier to encode and which can replace $\expneg$. Formally, we say that a distribution $d\in\mathcal{D}$ has a {\em bounded support}, if there exists $N(d)\geq 0$ such that $\mathbb{P}_{X\sim d}[|X|>N(d)]=0$. Here, we use $\mathbb{P}_{X\sim d}$ to denote the probability measure induced by a random variable $X$ with the probability distribution $d$. We say that a LinPP has the {\em bounded support property (BSP)} if all distributions in the program have bounded support. For instance, the program in Fig.~\ref{fig:intromot:b} has the BSP, whereas the program in Fig.~\ref{fig:intromot:a} does not.
%
%
Using the same notation as in Definition~\ref{def:genlexrsm-map}, we put:
\begin{equation*}
\weakexpneg(\boldsymbol{\lem},\tau) \equiv \mathbf{x}\in I(\loc)\cap\guards(\tau) \Rightarrow \forall 1 \leq j \leq \mathsf{lev}(\tau)\; \text{min-pre}_{\lem_j}^\tau(\loc,\vec{x}) \geq 0.
\end{equation*}
(The 'W' stands for ``weak.'')
Intuitively, $\expneg$ requires nonnegativity of the expected value of $\lem_j$ when integrated over successor states of level smaller than $j$, whereas the condition $\weakexpneg$ requires nonnegativity of the expected value of $\lem_j$ when integrated over all successor states. Since $\lem_j$ is nonnegative at successor states of level at least $j$, this new condition is weaker than $\expneg$. Nevertheless, the following lemma shows that in order to decide existence of LinGLexRSM maps for programs with the BSP, we may w.l.o.g.~replace $\expneg$ by $\weakexpneg$ for all transitions but for those of probabilistic branching. The proof of the lemma is deferred to Appendix~\ref{sec:appboundedtechnical}.


\begin{lemma}\label{lemma:technical}
Let $\pCFG$ be a LinPP with the BSP and $I$ be a linear invariant in $\pCFG$. If a LEM $\boldsymbol{\lem}$ satisfies conditions $\prank$ and $\pnneg$ for all transitions, $\expneg$ for all transitions in $\Delta_{PB}$ and $\weakexpneg$ for all other transitions, then $\boldsymbol{\lem}$ may be increased pointwise by a constant value in order to obtain a LinGLexRSM map.
\end{lemma}

\noindent{\em Algorithmic Results.} Let $\textsc{LinGLexPP}^{\textsc{bounded}}$ be the set of pairs $(\pCFG,I)$ of a pCFG $\pCFG$ representing a LinPP with the BSP and a linear invariant $I$ in $\pCFG$, such that $\pCFG$ admits a LinGLexRSM map supported by $I$.

\begin{theorem}
\label{thm:algo}
There is a polynomial-time algorithm deciding if a tuple $ (\pCFG,I) $ belongs to $\textsc{LinGLexPP}^{\textsc{bounded}}$. Moreover, if the answer is yes, the algorithm outputs a witness in the form of a LinGLexRSM map of minimal dimension.
\end{theorem}

 The algorithm behind Theorem~\ref{thm:algo} is a generalization of algorithms in~\cite{ADFG10:lexicographic,AgrawalC018} finding LinLexRFs in non-probabilistic programs and LinLexRSM maps in PPs, respectively. Suppose that we are given a LinPP $\pCFG=(\locs,\pvars,\gentransitions,\updates,\guards)$ with the BSP and a linear invariant $I$. Our algorithm stores a set $\mathcal{T}$ initialized to all transitions in $\pCFG$. It then proceeds in iterations to compute new components of the witness. In each iteration it searches for a LEM $\lem$ which is required to
\begin{compactenum}
	\item be nonnegative on each $\tau=(\loc,\delta)\in\mathcal{T}$, i.e.\ $\forall \mathbf{x}.\, \mathbf{x}\in I(\loc)\cap G(\tau) \Rightarrow \lem(\loc,\mathbf{x})\geq 0$;
	\item be unaffecting on each $\tau=(\loc,\delta)\in\mathcal{T}$, i.e.\ $\forall \mathbf{x}.\, \mathbf{x}\in I(\loc)\cap G(\tau) \Rightarrow \text{max-pre}_{\lem}^\tau(\loc,\vec{x})$ $\leq \lem(\loc,\mathbf{x})$;
	\item have nonnegative minimal pre-expectation for each $\tau=(\loc,\delta)\in\mathcal{T}\, \backslash \Delta_{PB}$, i.e.\ $\forall \mathbf{x}.\, \mathbf{x}\in I(\loc)\cap G(\tau) \Rightarrow \text{min-pre}_{\lem}^\tau(\loc,\vec{x}) \geq 0$;
	\item if $S$ is the set of states in $\pCFG$ whose all enabled transitions have been removed from $\mathcal{T}$ in the previous algorithm iterations, $\forall\tau=(\loc,\delta)\in\mathcal{T}\cap\Delta_{PB}$, $\forall \mathbf{x}.\, \mathbf{x}\in I(\loc)\cap G(\tau) \Rightarrow \text{pre}_{\lem,S}^\tau(\loc,\vec{x}) \geq 0$; and
	\item $1$-rank the maximal number of transitions in $\tau\in\mathcal{T}$, i.e.\ $\forall \mathbf{x}.\, \mathbf{x}\in I(\loc)\cap G(\tau) \Rightarrow \text{max-pre}_{\lem}^\tau(\loc,\vec{x}) \leq \lem(\loc,\mathbf{x})-1$ for as many $\tau=(\loc,\delta)$ as possible.
\end{compactenum}
This is done by fixing an LEM template for each location $\loc$ in $\pCFG$, and converting the above constraints to an equivalent linear program $\mathcal{LP}_{\mathcal{T}}$ in template variables via Farkas' lemma (FL). The FL conversion (and its extension to strict inequalities~\cite{CFNH16:prob-termination}) is
standard in termination proving and encoding conditions 1-3 and 5 above is analogous to~\cite{ADFG10:lexicographic,AgrawalC018}, hence we omit the details.
We show how condition 4 can be encoded via linear constraints in Appendix~\ref{app:algo}, along with the algorithm pseudocode and the proof of its correctness. In each algorithm iteration, all transitions that have been $1$-ranked are removed from $\mathcal{T}$ and the algorithm proceeds to the next iteration. If all transitions are removed from $\mathcal{T}$, the algorithm concludes that the program admits a LinGLexRSM map (obtained by increasing the constructed LEM by a constant defined in the proof of Lemma~\ref{lemma:technical}). If in some iteration a new component which $1$-ranks at least $1$ transition in $\mathcal{T}$ cannot be found, the program does not admit a LinGLexRSM map.

%

We conclude by showing that our motivating example in Fig.\ \ref{fig:intromot:b} admits a LinGLexRSM map supported by a very simple linear invariant. Thus, by completeness, our algorithm is able to prove its a.s.\ termination.

\begin{example}
	\label{ex:motivation-2}
	Consider the program in Figure~\ref{fig:intromot:b} with a linear invariant $I(\loc_0)=\textit{true}$, $I(\loc_1)=x\geq -7$. Its a.s.~termination is witnessed by a LEM $ \boldsymbol{\lem}(\loc_0,(x,y)) = (1,x+7,y+7) $, $ \boldsymbol{\lem}(\loc_1,(x,y)) = (1,x+8,y+7) $ and $ \boldsymbol{\lem}(\locterm,(x,y)) = (0,x+7,y+7) $. Since $ \Delta_{PB}=\emptyset $ here, and since $\prank$, $\pnneg$ and $\weakexpneg$ are satisfied by $\boldsymbol{\lem}$, by Lemma~\ref{lemma:technical}, $\pCFG$ admits a LinGLexRSM map supported by $I$.
\end{example}

\subsection{Algorithm for general LinPPs}

While imposing $\weakexpneg$ lets us avoid integration in LinPPs with the BSP, this is no longer the case if we discard the BSP.

Intuitively, the problem in imposing the condition $\weakexpneg$ instead of $\expneg$ for LinPPs without the BSP, is that the set of states of smaller level over which $\expneg$ performs integration might have a very small probability, however the value of the LinGLexRSM component on that set is negative and arbitrarily large in absolute value. Thus, a naive solution for general LinPPs would be to ``cut off'' the tail events where the LinGLexRSM component can become arbitrarily negative and over-approximate them by a constant value in order to obtain a piecewise linear GLexRSM map.
However, this might lead to the jump in maximal pre-expectation and could violate $\prank$.

In what follows, we consider a slight restriction on the syntax of LinPPs that we consider, and introduce a new condition on LEMs that allows the over-approximation trick mentioned above while ensuring that the $\prank$ condition is not violated. We consider the subclass LinPP$^{\ast}$ of LinPPs in which no transition of probabilistic branching and a transition with a sampling instruction share a target location. This is a very mild restriction (satisfied, e.g. by our motivating example in Fig.~\ref{fig:intromot:b}) which is enforced for technical reasons arising in the proof of Lemma~\ref{lemma:technical2}. Each LinPP can be converted to satisfy this property by adding a \textbf{skip} instruction in the program's source code where necessary. Second, using the notation of Definition~\ref{def:genlexrsm-map}, we define the new condition $\unbounded$ as follows:
\begin{equation*}
\begin{split}
    &\unbounded(\boldsymbol{\lem},\tau) \equiv \text{ if $\updates(\tau)=(i,u)$ with $u$ containing a sampling from a dis-}\\
    &\hspace{0.5cm}\text{ tribution of unbounded support, and $\loc'$ is the target location of $\tau$, then the }\\
    &\hspace{0.5cm}\text{ coefficient of the variable with index $i$ in $\lem_j(\loc')$}\text{ is $0$ for all $1\leq j<\mathsf{lev}(\tau)$ }.
\end{split}
\end{equation*}
The following technical lemma is an essential ingredient in the soundness proof of our algorithm for programs in LinPP$^{\ast}$. Its proof can be found in Appendix~\ref{app:technical2}.

\newlength{\textfloatseptemp}
\setlength{\textfloatseptemp}{\textfloatsep}
\setlength{\textfloatsep}{\baselineskip}


\begin{lemma}\label{lemma:technical2}
Let $\pCFG$ be a LinPP$^{\ast}$ and $I$ be a linear invariant in $\pCFG$. If a LEM $\boldsymbol{\lem}$ satisfies $\prank$ and $\pnneg$ for all transitions, $\expneg$ for all transitions of probabilistic branching, $\weakexpneg$ for all other transitions, as well as $\unbounded$, then $\pCFG$ admits a piecewise linear GLexRSM map supported by $I$.
\end{lemma}

\noindent{\em Algorithm.} The new algorithm shares an overall structure with the algorithm from Section~\ref{sec:boundeddis}. Thus, we only give a high level overview and focus on novel aspects. The algorithm pseudocode is presented in Algorithm~\ref{algo:generalprogs}.

\begin{algorithm}[t]

	\SetKwInOut{Input}{input}\SetKwInOut{Output}{output}
	\DontPrintSemicolon

	\Input{A LinPP$^{\ast}$ $\pCFG$, linear invariant $I$.}
	\Output{An LEM satisfying the conditions of Lemma~\ref{lemma:technical2}, if it exists}
	$\mathcal{T}$ $\longleftarrow$ all transitions in $\pCFG$; $d$ $\longleftarrow$ $0$\\
	\While{$\mathcal{T}$ is non-empty}{
		construct $\mathcal{LP}_{\mathcal{T}}^{\text{unb}}$\label{aline:reduce}\\
		\If{$\mathcal{LP}_{\mathcal{T}}^{\text{unb}}$ is feasible}{
		    $d$ $\longleftarrow$ $d+1$; $\lem_d$ $\longleftarrow$ LEM defined by the optimal solution of $\mathcal{LP}_{\mathcal{T}}^{\text{unb}}$\\
			$\mathcal{T}$ $\longleftarrow$ $\mathcal{T}\backslash \{\tau\in\mathcal{T}\mid \tau \text{ is 1-ranked by } \lem_d\}$}\label{aline:prune}
		\Else{
		    found $\longleftarrow$ false\\
		    \For{$\tau_0\in\transitions^{\text{unb}}\cap\,\mathcal{T}$}{
		        construct $\mathcal{LP}_{\mathcal{T}}^{\tau_0,\text{unb}}$\\
		        \If{$\mathcal{LP}_{\mathcal{T}}^{\tau_0,\text{unb}}$ is feasible}{
		            $d$ $\longleftarrow$ $d+1$; found $\longleftarrow$ true\\
			        $\lem_d$ $\longleftarrow$ LEM defined by the optimal solution of $\mathcal{LP}_{\mathcal{T}}^{\tau_0,\text{unb}}$\\
			        $\mathcal{T}$ $\longleftarrow$ $\mathcal{T}\backslash \{\tau\in\mathcal{T}\mid \tau \text{ is 1-ranked by } \lem_d\}$}}
			 \lIf{not found}{\Return No LEM as in Lemma~\ref{lemma:technical2}}}
	}
	\Return{$(\lem_1,\dots,\lem_d)$}
	\caption{Algorithm for proving a.s.~termination in LinPP$^{\ast}$.}
	\label{algo:generalprogs}
\end{algorithm}

The condition $\unbounded$ is encoded by modifying the templates for the new LEM components. Let $\mathord{\mapsto}^{\text{unb}}$ be the set of transitions in $\pCFG$ containing sampling from unbounded support distributions, and for any such transition $\tau$ let $\loc'_{\tau}$ be its target location. Then for any set of transitions $\mathcal{T}$, construct a linear program $\mathcal{LP}_{\mathcal{T}}^{\text{unb}}$ analogously to $\mathcal{LP}_{\mathcal{T}}$ in Section~\ref{sec:boundeddis}, additionally enforcing that for each $\tau\in\mathord{\mapsto}^{\text{unb}}\cap\,\mathcal{T}$, the coefficient of the variable updated by $\tau$ in the LEM template at $\loc'_{\tau}$ is $0$. Algorithm~\ref{algo:generalprogs} first tries to prune as many transitions as possible by repeatedly solving $\mathcal{LP}_{\mathcal{T}}^{\text{unb}}$ and removing ranked transitions from $\mathcal{T}$, see lines 3-6. Once no more transitions can be ranked, the algorithm tries to rank new transitions by allowing non-zero template coefficients previously required to be $0$, while still enforcing $\unbounded$. For a set of transitions $\mathcal{T}$ and for $\tau_0\in \mathord{\mapsto}^{\text{unb}}\cap\,\mathcal{T}$, we construct a linear program $\mathcal{LP}_{\mathcal{T}}^{\tau_0,\text{unb}}$ analogously to $\mathcal{LP}_{\mathcal{T}}^{\text{unb}}$ but allowing a non-zero coefficient of the variable updated by $\tau_0$ at $\loc'_{\tau_0}$. However, we further impose that the new component $1$-ranks any other transition in $\mathord{\mapsto}^{\text{unb}}\cap\,\mathcal{T}$ with the target location $\loc'_{\tau_0}$. 
This new linear program is solved for all $\tau_0\in \mathord{\mapsto}^{\text{unb}}\cap\,\mathcal{T}$ and all $1$-ranked transitions are removed from $\mathcal{T}$, as in Algorithm~\ref{algo:generalprogs}, lines 7-15. The process continues until all transitions are pruned from $ \mathcal{T} $ or until no remaining transition can be 1-ranked, in which case no LEM as in Lemma~\ref{lemma:technical2} exists. 

\begin{theorem}\label{thm:soundness2}
Algorithm~\ref{algo:generalprogs} decides in polynomial time if a LinPP$^{\ast}$ $\pCFG$ admits an LEM which satisfies all conditions of Lemma~\ref{lemma:technical2} and which is supported by $I$. Thus, if the algorithm outputs an LEM, then $\pCFG$ is a.s.\ terminating and admits a piecewise linear GLexRSM map supported by $I$.
\end{theorem}

The proof of Theorem~\ref{thm:soundness2} can be found in Appendix~\ref{app:soundess2}. We conclude by showing that Algorithm~\ref{algo:generalprogs} can prove a.s.~termination of our motivating example in Fig.~\ref{fig:intromot:a}.

\begin{example}
	\label{ex:mot1} Consider the program in Figure~\ref{fig:intromot:a} with a linear invariant $I(\loc_0)=\textit{true}$, $I(\loc_1)=y\geq 0$. The LEM defined via $ \boldsymbol{\lem}(\loc_0,(x,y)) = (1,2y+2,x+1)$, $ \boldsymbol{\lem}(\loc_1,(x,y)) = (1,2y+1,x+1)$ and $ \boldsymbol{\lem}(\locterm,(x,y)) = (0,2y+2,x+1) $ satisfies $\prank$, $\pnneg$ and $\weakexpneg$, which is easy to check. Furthermore, the only transition containing a sampling instruction is the self-loop at $\loc_1$ which is ranked by the third component of $\boldsymbol{\lem}$. As the coefficients of $x$ of the first two components at $\loc_1$ are equal to $0$, $\boldsymbol{\lem}$ also satisfies $\unbounded$. Hence, $\boldsymbol{\lem}$ satisfies all conditions of Lemma~\ref{lemma:technical2} and Algorithm~\ref{algo:generalprogs} proves a.s.~termination.
\end{example}

	\section{Conclusion}

	In this work we present new lexicographic termination certificates for probabilistic programs. We also show how to automate the search for the new certificate within a wide class of probabilistic programs.
	An interesting direction of future work would be automation beyond linear arithmetic programs.

	\section*{Acknowledgements}
	
	This research was partially supported by the ERC CoG 863818 (ForM-SMArt), the Czech Science Foundation grant No.~GJ19-15134Y, and the European Union’s Horizon 2020 research and innovation programme under the Marie Skłodowska-Curie Grant Agreement No.~665385.

\bibliography{popl20}

\clearpage
\appendix

\begin{center}
	{\Large Appendix}
\end{center}

\section{Missing Parts in the Proof of Theorem~\ref{thm:genlexrsm}}
\label{app:genlglexrsm}

\noindent{\em Claim.} We first prove the claim at the beginning of the proof, i.e.~that there exist $1\leq k\leq n$ and $s,M\in\mathbb{N}_0$ such that the set $B$ of all $\omega\in \Omega$ for which the following properties hold has positive measure, i.e.~$\mathbb{P}[B]>0$: (1)~$T(\omega)=\infty$, (2)~$\mathbf{X}_s[k](\omega)\leq M$, (3)~for each $t\geq s$, the level of $\omega$ at step $t$ is at least $k$, and (4) the level of $\omega$ equals $k$ infinitely many times.

\smallskip For $\omega\in \Omega$, we define $\mathsf{minlev}(\omega)$ to be the smallest $0\leq j\leq n$ such that the level of $\omega$ is equal to $j$ in infinitely many steps $i$. Then to prove the Claim, let $B_k =\{\omega\in\Omega \mid T(\omega)=\infty \land \minlev(\omega)=k\}$ for each $1\leq k\leq n$. Then
\[ \{\omega\in\Omega \mid T(\omega)=\infty\} = \cup_{k=1}^n B_k, \]
and thus, by the union bound, there exists $1\leq k\leq n$ for which $\mathbb{P}[B_k]>0$. We may express $B_k$ as a union of events over the time of the last visit to some $L^i_j$ with $j<k$. If we write $B_k^s = \{\omega \in \Omega \mid T(\omega) = \infty \land (i\geq s\Rightarrow \omega \in \cup_{j=k}^n L^i_{j}\}$ for each $s\in\mathbb{N}_0$, we have that $B_k = \cup_{s\geq 0}B_k^s$. As by the union bound $\mathbb{P}[B_k]\leq \sum_{s=0}^{\infty}\mathbb{P}[B^s_k]$, there exists $s\in\mathbb{N}_0$ for which $\mathbb{P}[B_k^s]>0$. Now, for each $M\in\mathbb{N}_0$, let $B_k^{s,M}$ be defined via
\begin{equation*}
B_k^{s,M} = \{\omega\in B_k^s \mid \mathbf{X}_s[k]\leq M\}.
\end{equation*}
Then $B_k^s=\cup_{M=0}^{\infty}B_k^{s,M}$. By the union bound we have $\mathbb{P}[B_k^s]\leq\sum_{M=0}^{\infty}\mathbb{P}[B_k^{s,M}]$, and there exists $M\in\mathbb{N}_0$ such that $\mathbb{P}[B_k^{s,M}]>0$. The set $B=B_k^{s,M}$ satisfies the conditions of the Claim.

\medskip\noindent{\em Inequality} Now we derive the inequality that was used in the proof but whose proof was deferred to the Appendix, i.e.~that for each $t\geq s$ we have
\begin{equation*}
\mathbb{E}[Y_{t+1}] \leq \mathbb{E}[Y_t] - \mathbb{P}[L^t_k\cap D\cap\{F>t\}],
\end{equation*}
where notation is as in the proof of Theorem~\ref{thm:genlexrsm}.

\medskip\noindent To see this, note that
\begin{equation}\label{eq:sum}
\mathbb{E}[Y_{t+1}] = \mathbb{E}[Y_{t+1}\cdot \mathbb{I}(\Omega\backslash D)] + \mathbb{E}[Y_{t+1}\cdot \mathbb{I}(D\cap \{F\leq t\})] + \mathbb{E}[Y_{t+1}\cdot \mathbb{I}(D\cap \{F>t\})],
\end{equation}
where
\begin{itemize}
	\item $\mathbb{E}[Y_{t+1}\cdot \mathbb{I}(\Omega\backslash D)]=0$ since $Y_{t+1}=0$ on $\Omega\backslash D$;
	\item $\mathbb{E}[Y_{t+1}\cdot \mathbb{I}(D\cap \{F\leq t\})] = \mathbb{E}[Y_t\cdot \mathbb{I}(D\cap \{F\leq t\})]$, as $Y_{t+1}=Y_t$ on $D\cap \{F\leq t\}$;
	\item $\mathbb{E}[Y_{t+1}\cdot \mathbb{I}(D\cap \{F>t\})] =\mathbb{E}[\mathbf{X}_{t+1}[k]\cdot \mathbb{I}(D\cap \{F>t\})]$, as $Y_{t+1}=\mathbf{X}_{t+1}[k]$ on $D\cap \{F>t\}$.
\end{itemize}

\medskip\noindent Now, by the theorem assumptions the conditional expectation $\mathbb{E}[\mathbf{X}_{t+1}[k]\mid \mathcal{F}_t]$ exists. Hence, as $D\cap \{F>t\}$ is $\mathcal{F}_t$-measurable for $t\geq s$, we have that
\begin{equation*}
\mathbb{E}\Big[\mathbf{X}_{t+1}[k]\cdot \mathbb{I}(D\cap \{F>t\})\Big] = \mathbb{E}\Big[\mathbb{E}[\mathbf{X}_{t+1}[k]\mid \mathcal{F}_t]\cdot \mathbb{I}(D\cap \{F>t\})\Big].
\end{equation*}
By the theorem assumptions, we also know that $\mathbb{E}[\mathbf{X}_{t+1}[k]\mid \mathcal{F}_t]\leq \mathbf{X}_t[k]-\mathbb{I}(L^t_k)$ on $\cup_{j=k}^n L^t_{j}$. Thus, as  $D\cap\{F>t\}\subseteq \cup_{j=k}^n L^t_{j}$ for $t\geq s$, by plugging this into the previous equation we conclude
\begin{equation*}
\begin{split}
\mathbb{E}\Big[\mathbf{X}_{t+1}[k]\cdot \mathbb{I}(D\cap \{F>t\})\Big] &\leq  \mathbb{E}\Big[(\mathbf{X}_t[k]-\mathbb{I}(L^t_k))\cdot \mathbb{I}(D\cap \{F>t\})\Big] \\
&= \mathbb{E}\Big[(Y_t-\mathbb{I}(L^t_k))\cdot \mathbb{I}(D\cap \{F>t\}\Big],
\end{split}
\end{equation*}
where in the last row we used that $\mathbf{X}_t[k]=Y_t$ on $D\cap\{F>t\}$.

\medskip\noindent Summing up our upper bounds on each term on the RHS of~\eqref{eq:sum}, we conclude that
\begin{equation*}
\begin{split}
\mathbb{E}[Y_{t+1}] &\leq 0 + \mathbb{E}[Y_t\cdot \mathbb{I}(D\cap \{F\leq t\})] + \mathbb{E}\Big[(Y_t-\mathbb{I}(L^t_k))\cdot \mathbb{I}(D\cap \{F>t\})\Big] \\
&= \mathbb{E}[Y_t\cdot \mathbb{I}(D)] - \mathbb{P}[L^t_k\cap D\cap \{F>t\}], \\
&= \mathbb{E}[Y_t] - \mathbb{P}[L^t_k\cap D\cap \{F>t\}].
\end{split}
\end{equation*}
where in the last row we used that $Y_t=0$ on $\Omega\backslash D$. This proves the desired inequality.

\section{Defining Successor States}
\label{app:succstates}

We say that $ (\loc',\vec{x}') $ is a successor of $ (\loc,\vec{x}) $ under transition $ \tau $ if $ \loc' $ is the successor location of $ \loc $ under $ \tau $ and $ \vec{x}' $ relates to $ \vec{x} $ in one of these ways, depending on $ (i,u) = \updates(\tran) $:


%
\begin{itemize}
	\item If $\up = \bot$, then $ \vec{x}' = \vec{x} $.
	\item If $\up$ is an expression, then for the sampling instruction in the expression (if it exists) we sample from the respective distribution and replace the instruction with the sampled value. We then evaluate the resulting expression into a number $ A $ an put $ \vec{x}' = \vec{x}(i \leftarrow A) $.
	\item If $ \up $ is an interval $ R $, then we sample a value $A$ is from the distribution prescribed by $\sigma$ for $\genpath_i$ and $\tau$; we then put $ \vec{x}' = \vec{x}(i \leftarrow A) $.
\end{itemize}

\section{Generalized Definition of Pre-expectation}\label{app:preexpectation}

	In what follows, we generalize the definition presented in~\cite{SriramCAV} in order to allow taking expectation over subsets of states in $\pCFG$ (a necessity for handling the $ \expneg $ constraint). We say that a set $S$ of states in $\pCFG$ is {\em measurable}, if for each location $\loc$ in $\pCFG$ we have that $\{\mathbf{x}\in\mathbb{R}^{|\vars|}\mid (\loc,\mathbf{x})\in S\}\in \mathcal{B}(\mathbb{R}^{|\vars|})$, i.e.\ it is in the Borel sigma-algebra of $\mathbb{R}^{|\vars|}$. Furthermore, we also differentiate between the {\em maximal} and {\em minimal pre-expectation}, which may differ in the case of non-deterministic assignments in programs and intuitively are equal to the maximal resp.~minimal value of the next-step expectation.
	
	\begin{definition}[Pre-expectation]
	\label{def:preexp}
	Let $\lem$ be a 1-dimensional MM, $ \tau = (\loc, \succDist) $ a transition and $S$ be a measurable set of states in $\pCFG$.
	A \emph{maximal pre-expectation} of $ \lem $ in $ \tau $ given $S$ is the function $ \text{max-pre}_{\lem,S}^\tau $ assigning to each state $ (\loc,\vec{x}) $ the following number:
	\begin{itemize}
	\item if $\tau\in\gentransitions_{PB}$ is a transition of probabilistic branching, then $ \text{max-pre}_{\lem,S}^\tau(\loc,\vec{x}) = \sum_{\loc' \in \locs} \succDist(\loc')\cdot \lem(\loc',\vec{x})\cdot \mathbb{I}(S)(\loc',\mathbf{x}) $;
	\item otherwise $ \text{max-pre}_{\lem,S}^\tau(\loc,\vec{x}) =  A_{\loc'}$ where $\loc'$ is the unique successor location of $\tau$ 
	and $ A_{\loc'} $ is defined as follows:
	\begin{enumerate}
	\item if $\updates(\tau)=\bot$, then $A_{\loc'}=\lem(\loc',\mathbf{x})\cdot \mathbb{I}(S)(\loc',\mathbf{x})$
	\item if $\updates(\tau)=(j,\up)$ and $ \up $ is an expression over program variables and sampling instructions, then $ A_{\loc'}= \lem(\loc',\vec{x}(j\leftarrow A))$, where $ A $ is the expected value of $ \up\cdot \mathbb{I}(S)(\loc',u) $ under valuation $ \vec{x} $ (i.e.~integration of $u$ is performed only over the set of states at $\loc'$ that belong to $S$; the integral is well-defined and finite since $\{\mathbf{x}\in\mathbb{R}^{|\vars|}\mid (\loc',\mathbf{x})\in S\}\in \mathcal{B}(\mathbb{R}^{|\vars|})$ and any distribution appearing in variable updates is integrable); 
	\item if $\updates(\tau)=(j,\up)$ and $ \up $ is an interval $ I $ denoting a nondeterministic assignment, then $ A_{\loc'} = \sup_{y \in I\land (\loc',\vec{x}(j\leftarrow 
	y))\in S}  \lem(\loc',\vec{x}(j\leftarrow 
	y)) $.
	\end{enumerate}
	\end{itemize}
	A \emph{minimal pre-expectation} is denoted by $ \text{min-pre}_{\lem,S}^\tau $ and is defined analogously, with the only difference being that in the point 3 we use $ \inf $ instead of $ \sup $.
	We omit the subscript $S$ when $ S $ is the set of all states of $ \pCFG $.
	\end{definition}

\section{Soundness Proof for the GLexRSM maps (Theorem~\ref{thm:genlexrsmmap})}
\label{app:map-soundness}

Let $\pCFG=(\locs,\pvars,\gentransitions,\updates,\guards)$ be a pCFG. The proof proceeds by using the GlexRSM map $\boldsymbol{\lem}$ to define a GLexRSM w.r.t.~$\Term$ in the probability space $(\Omega_\pCFG,\mathcal{F}_\pCFG,\mathbb{P})$ associated to the pCFG. Here, $ \probm $ is a probability measure defined by an arbitrary (but throughout the proof fixed) scheduler and an initial configuration. The proof then uses Theorem~\ref{thm:genlexrsm} to conclude a.s.~termination of the program. 

First, however, we need to define the auxiliary notion of \emph{pre-expectation with respect to a scheduler.}

\begin{definition}[Pre-expectation with respect to a scheduler]
	\label{def:preexpsch}
	Let $S$ be a set of measurable states in $\pCFG$. The \emph{pre-expectation with respect to a scheduler $\sigma$ of a MM $ \lem $} given $S$ is a map $\text{pre}_{\sigma,\lem,S}:\Fpath_\pCFG\rightarrow \mathbb{R}$
	defined as follows. Let $\rho\in\Fpath_\pCFG$ be a finite path ending in $(\loc,\mathbf{x})$. Let $\sigma(\rho)$ denote the distribution over transitions enabled in $(\loc,\mathbf{x})$ defined by the scheduler, and for each transition $ \tau = (\loc, \succDist) \in \support(\sigma(\rho))$ with update element being a nondeterministic assignment from an interval, let $d_\sigma(\rho,\tau)$ be a distribution over the interval defined by $ \sigma $. Then
	\[ \text{pre}_{\sigma,\lem,S}(\rho)=\sum_{\tau=(l,\delta)\in\support(\sigma(\rho))}\sigma(\rho)(\tau)\cdot \text{pre}_{\sigma,\lem,S}^{\tau}(\loc,\vec{x}), \]
	where $ \text{pre}_{\sigma,\lem,S}^{\tau} $ is defined in the same way as the standard pre-expectation $ \text{pre}_{\lem,S}^{\tau} $, \emph{except for the case} when $ \tau = (\loc,\delta)$ carries a non-deterministic assignment, i.e. $ \up(\tau) = (j,u) $ with $ u $ being an interval; in such a case, we put $ \text{pre}_{\sigma,\lem}^{\tau}(\loc,\vec{x}) =  \sum_{\loc' \in \locs} \succDist(\loc')\cdot \lem(\loc',\vec{x}(j\leftarrow \E[d_{\sigma}(\rho,\tau)\cdot \mathbb{I}(S)])$.
	%
	%
	%
\end{definition}

Intuitively, $\text{pre}_{\sigma,\lem,S}$ takes a finite path $\rho\in\Fpath_\pCFG$ as an input, and 
returns the {\em expected value} of $\lem$ in the next step when the integration is performed over program states in $S$, given the program run history and the choices of the scheduler $ \sigma $.



For a run $\rho\in \Omega_{\pCFG}$, let $(\loc^{\rho}_i,\mathbf{x}^{\rho}_i)$ denote the $i$-th configuration along $\rho$, $\tau^{\rho}_i$ the $i$-th transition taken along $\rho$ and $\rho_i$ the prefix of $\rho$ of length $i$. We define an $n$-dimensional stochastic process $(\mathbf{X}_i)_{i=0}^{\infty}$ over $(\Omega_\pCFG,\mathcal{F}_\pCFG,\mathbb{P})$ by setting
\begin{equation*}
\mathbf{X}_i[j](\rho) =
\begin{cases}
\boldsymbol{\eta}_j(\loc^{\rho}_i,\mathbf{x}^{\rho}_i) &\text{if } \loc^{\rho}_i \neq \locterm\\
-1 &\text{otherwise}.
\end{cases}
\end{equation*}
for each $i\in\mathbb{N}_0$, $1\leq j\leq n$ and $\rho\in\Omega_{\pCFG}$. We consider the canonical filtration $(\mathcal{R}_i)_{i=0}^{\infty}$ where $\natfilt_i$ is the smallest sub-sigma-algebra of 
$\mathcal{F}_\pCFG$ such that all the functions $\loc_j^{\cdot}(\rho)=\loc_j^{\rho}$, $ \vec{x}_j^{\cdot}(\rho)=\vec{x}_j^{\rho} $, $0\leq j \leq i$, are $\natfilt_i$-measurable. We also consider the stopping time $\ttime$ with respect to $(\mathcal{R}_i)_{i=0}^{\infty}$ defined by the first hitting time of $\locterm$. For each $i\in\mathbb{N}_0$ we define a partition of $\{\Term>i\}$ into $n$ sets $L^i_1,\dots,L^i_n$ with 
$L^i_j = \{\rho\in\Omega_{\pCFG}\mid \loc^{\rho}_i\neq \locterm,\, \text{ level of } (\loc^{\rho}_i,\mathbf{x}^{\rho}_i) \text{ is } j \}.$
We show that $(\{\vec{X}_{i=0}^{\infty},\{L_1^i,\dots,L_n^i\}_{i=0}^{\infty})$ is an instance of a GLexRSM in the probability space $(\Omega_\pCFG,\mathcal{F}_\pCFG,\mathbb{P})$. We prove this by verifying each of the defining conditions in Definition~\ref{def:genlexrsm}:
\begin{enumerate}
	\item Clearly, each $\mathbf{X}_i[j]$ is $\mathcal{R}_i$-measurable as $\mathbf{X}_i[j]$ is defined in terms of the $i$-th configuration of a program run.
	\item We now show that all conditional expectations required in Definition~\ref{def:genlexrsm} exist. Let $B\in \mathcal{R}_{i+1}$, we need to show that $\mathbb{E}[\mathbf{X}_{i+1}[j] \cdot \mathbb{I}(B)\mid \mathcal{R}_i]$ exists. Fix $\rho\in \Omega_{\pCFG}$, and let $\rho_i$ denote a finite prefix of $\rho$ of length $i$. Then define $\text{proj}(B)(\rho_i)$ as
	\[ \text{proj}(B)(\rho_i) = \{(\loc^{\rho'}_{i+1},\mathbf{x}^{\rho'}_{i+1}) \mid \rho'\in B \land \rho'_i=\rho_i\},\]
	i.e.~it is the set of all $(i+1)$-st states of infinite runs in $B$ whose finite prefix of length $i$ coincides with $\rho_i$. Then we show that
    \[\mathbb{E}[\mathbf{X}_{i+1}[j] \cdot \mathbb{I}(B)\mid \mathcal{R}_i](\rho) = \text{pre}_{\sigma,\boldsymbol{\lem}_j,\text{proj}(B)(\rho_i)}(\rho_i),\]
    i.e.~this conditional expectation can be expressed in terms of the pre-expectation w.r.t.~the scheduler that we define in Definition~\ref{def:preexpsch}. 
    
    To see this, recall that the conditional expectation $\mathbb{E}[\mathbf{X}_{i+1}[j] \cdot \mathbb{I}(B)\mid \mathcal{R}_i]$ is defined as the unique $\mathcal{R}_i$-measurable random variable $Y$ for which $\mathbb{E}[Y\cdot \mathbb{I}(A)]=\mathbb{E}[\mathbf{X}_{i+1}[j] \cdot \mathbb{I}(B\cap A)]$ for any $A\in\mathcal{R}_i$. For fixed $A\in\mathcal{R}_i$, $\mathbb{E}[\mathbf{X}_{i+1}[j] \cdot \mathbb{I}(B\cap A)]$ is equal to the integral of $\mathbf{X}_{i+1}[j]$ when integration is done over all runs in $B$ whose finite prefix of length $i$ ensures that the run belongs to $A$. But this is exactly the value obtained by integrating $\text{pre}_{\sigma,\boldsymbol{\lem}_j,\text{proj}(B)(\rho_i)}(\rho_i)$ over runs $\rho\in A$. This proves that $\text{pre}_{\sigma,\boldsymbol{\lem}_j,\text{proj}(B)(\rho_i)}(\rho_i)$ satisfies all the properties of the conditional expectation for any $\rho\in \Omega_{\pCFG}$. As conditional expectation is a.s.~unique whenever it exists~\cite{Williams:book}, the claim follows.
	\item To see that lexicographic ranking, nonnegativity and boundedness of conditional expectation conditions in Definition~\ref{def:genlexrsm} hold, we need to show that for each $i\in\mathbb{N}_0$, $1\leq j\leq n$ and $\rho \in L^i_j$, we have
	\begin{itemize}
	    \item $\mathbb{E}[\mathbf{X}_{i+1}[j'] \mid \mathcal{F}_i](\rho) \leq \mathbf{X}_i[j'](\rho)$ for each $1\leq j'<j$,
		\item $\mathbb{E}[\mathbf{X}_{i+1}[j] \mid \mathcal{F}_i](\rho) \leq \mathbf{X}_i[j](\rho) - 1$,
		\item $\mathbf{X}_i[j'](\rho)\geq 0$ for each $1\leq j'\leq j$, and
		\item $\mathbb{E}[\mathbf{X}_{i+1}[j']\cdot \mathbb{I}(\cup_{t=0}^{j'-1}L^{i+1}_{t}) \mid \mathcal{F}_i](\rho)\geq 0$ for each $1\leq j'\leq j$ (here, $L^{i+1}_0=\{\Term\leq i+1\}$).
	\end{itemize}
	Before checking these properties, we first prove in Proposition~\ref{poposition:detsch} below that, to show that a program terminates almost-surely, it suffices to consider those schedulers that to each finite program run assign a single transition to be taken. Thus, we without loss of generality assume that $\sigma$ is such a scheduler, and let $\tau$ be the only transition in the support of $\sigma(\rho_i)$. Then since $\rho\in L^i_j$, we have $j=\mathsf{lev}(\tau)$. As $\tau$ is the only transition in the support of $\sigma(\rho_i)$, we have
	\begin{equation*}
	\begin{split}
	\mathbb{E}[\mathbf{X}_{i+1}[j]\mid \mathcal{R}_i](\rho) &= \text{pre}_{\sigma,\boldsymbol{\eta_j}}(\rho_i) = \text{pre}_{\sigma,\boldsymbol{\eta_j}}^{\tau}(\loc^{\rho}_i,\mathbf{x}^{\rho}_i) \leq  \text{max-pre}_{\boldsymbol{\eta_j}}^{\tau}(\loc^{\rho}_i,\mathbf{x}^{\rho}_i)\\
	&\leq \boldsymbol{\eta}_j(\loc^{\rho}_i,\mathbf{x}^{\rho}_i)-1 = \mathbf{X}_i[j](\rho)-1,
	\end{split}
	\end{equation*}
	The inequality $\text{pre}_{\sigma,\boldsymbol{\eta_j}}^{\tau}(\loc^{\rho}_i,\mathbf{x}^{\rho}_i) \leq  \text{max-pre}_{\boldsymbol{\eta_j}}^{\tau}(\loc^{\rho}_i,\mathbf{x}^{\rho}_i)$ holds since $\text{pre}_{\sigma,\boldsymbol{\eta_j}}^{\tau}$ in Definition~\ref{def:preexpsch} and $\text{max-pre}_{\boldsymbol{\eta_j}}^{\tau}$ in Definition~\ref{def:preexp} only differ in the way in which we treat update elements given by nondeterministic assignments where in the first case we take expectation and in the second case the global maximum. The last inequality holds by the $\prank$ condition of GLexRSM maps and since $j=\mathsf{lev}(\tau)$. It follows analogously that $\mathbb{E}[\mathbf{X}_{i+1}[j'] \mid \mathcal{F}_i](\rho) \leq \mathbf{X}_i[j'](\rho)$ for each $1\leq j'<j$. The fact that $\mathbf{X}_i[j'](\rho)\geq 0$ for each $1\leq j'\leq j$ follows from the $\pnneg$ condition of GLexRSM maps. Finally, for each $1\leq j'\leq j$ we have that $\mathbb{E}[\mathbf{X}_{i+1}[j']\cdot \mathbb{I}(\cup_{t=0}^{j'-1}L^{i+1}_{t}) \mid \mathcal{F}_i](\rho) \geq \text{min-pre}_{\lem_{j'},S^{\leq j'-1}_{\mathsf{lev}}}^\tau(\loc^{\rho}_i,\vec{x}^{\rho}_i)$ as $\mathsf{lev}(\tau)=j$, so the expected leftward nonnegativity follows from the $\expneg$ condition of GLexRSM maps.
\end{enumerate}
Therefore, $((\mathbf{X}_i)_{i=0}^{\infty},(L^i_1,\dots,L^i_n)_{i=0}^{\infty})$ is an instance of a GLexRSM w.r.t. the stopping time $\Term$. Thus, and from Theorem~\ref{thm:genlexrsm} we conclude $\mathbb{P}[\Term<\infty]=1$.

\begin{proposition}\label{poposition:detsch}
	Let $\pCFG$ be a pCFG and suppose that there exist a measurable scheduler $\sigma$ and an initial configuration $(\locinit,\vecinit)$ such that $\mathbb{P}^{\sigma}_{\locinit,\vecinit}[\Term = \infty] > 0$. Then there exists a measurable scheduler $\sigma^{\ast}$ which is deterministic in the sense that to each finite path it assigns a single transition with probability $1$, such that $\mathbb{P}^{\sigma^{\ast}}_{\locinit,\vecinit}[\Term = \infty] > 0$.
\end{proposition}

\begin{proof}
	In what follows we fix the initial configuration $(\locinit,\vecinit)$ and omit it from the notation of the probability measure, which we denote by $\mathbb{P}^{\sigma}$. We construct $\sigma^{\ast}$ by constructing a sequence $\sigma=\sigma_0,\sigma_1,\sigma_2,\dots$ of schedulers, where for each $i\in\mathbb{N}_0$ we have that $\sigma_{i+1}$ and $\sigma_i$ agree on histories of length at most $i-1$, and $\sigma_{i+1}$ ''refines'' $\sigma_i$ on histories of length $i$ in such a way that:
	\begin{itemize}
		\item $\sigma_{i+1}$ is deterministic on histories of length at most $i$;
		\item $\sigma_{i+1}$ is measurable;
		\item $\mathbb{P}^{\sigma_{i+1}}[\Term=\infty] \geq \mathbb{P}^{\sigma_i}[\Term=\infty]$.
	\end{itemize}
	Then we define the scheduler $\sigma^{\ast}$ as $\sigma^{\ast}(\rho)=\sigma_i(\rho)$ whenever the length of a finite history $\rho$ is $i$.
	
	The construction of $\sigma_{i+1}$ from $\sigma_i$ proceeds as follows. For finite runs of length different than $i$, we let $\sigma_{i+1}(\rho)=\sigma_i(\rho)$. We are left to define $\sigma_{i+1}$ on histories of length exactly $i$. Fix a finite history $\rho$ in $\pCFG$ of length $i$. For each transition $\tau\in\support(\sigma_i(\rho))$, let $p_i(\rho,\tau)$ denote the probability of non-termination in the probability space of infinite runs starting in the last configuration of $\rho$ under the scheduler $\sigma_i$ (where the history $\rho$ is taken into account to resolve nondeterminism). Then set
	\[ \tau_i(\rho) = \arg \max_{\tau\in\support(\sigma_i(\rho))}p_i(\rho,\tau), \]
	i.e.~the transition in the support of $\sigma(\rho)$ which maximizes this probability. Then we define $\sigma_{i+1}(\rho)$ to be a Dirac-delta distribution which assigns probability $1$ to the transition $\tau_i(\rho)$.
	
	We now check that the constructed scheduler $\sigma_{i+1}$ satisfies each of the $3$ properties above:
	\begin{enumerate}
		\item The fact that, $\sigma_{i+1}$ is deterministic on histories of length at most $i$ immediately follows by induction on $i$.

		\item To see that $\sigma_{i+1}$ is measurable, we again proceed by induction on $i$ and assume that $\sigma_i$ is measurable (base case holds since $\sigma_0=\sigma$). We need to show that, for each program transition $\tau$, we have $\sigma_{i+1}^{-1}(\tau)\in \mathcal{F}_{\mathsf{fin}}$ where $\mathcal{F}_{\mathsf{fin}}$ is the $\sigma$-algebra of finite runs (for the definition of this $\sigma$-algebra, as well as for details on measurable schedulers, see~\cite{AgrawalC018}). Write
		\[ \sigma_{i+1}^{-1}(\tau) = A_{\tau,<i}\cup A_{\tau,=i}\cup A_{\tau,>i}, \]
		where $A_{\tau,<i}$ is the set of all finite runs in $\sigma_{i+1}^{-1}(\tau)$ of length at most $i-1$ and analogously for the other two sets. Since $\sigma_{i+1}$ and $\sigma_i$ coincide on histories of length different than $i$, by measurability of $\sigma_i$ it follows that $A_{\tau,<i}$ and $A_{\tau,>i}$ are in $\mathcal{F}_{\mathsf{fin}}$. Thus it suffices to show that $A_{\tau,=i}\in\mathcal{F}_{\mathsf{fin}}$.
		
		We partition $A_{\tau,=i}$ as $\cup_{\loc\in L}A_{\tau,=i,l}$, where $A_{\tau,=i,l}$ is the set of all finite runs in $A_{\tau,=i}$ with the last location being $l$. It suffices to prove that each $A_{\tau,=i,l}\in\mathcal{F}_{\mathsf{fin}}$.
		
		To show that $A_{\tau,=i,l}\in\mathcal{F}_{\mathsf{fin}}$, we again partition $A_{\tau,=i,l}$ in terms of transitions in the support of $\sigma_i$. Therefore, it suffices to prove for each finite set of transitions $T\subseteq\Delta$ with $\tau\in T$ that $A_{\tau,=i,l,T}\in\mathcal{F}$, where $A_{\tau,=i,l,T}$ is the set of all finite paths of length $i$ and ending in $l$ such that $\support(\sigma_i(\rho))=T$ and $p_i(\rho,\tau)\geq p_i(\rho,\tau')$ for any $\tau'\in T$.
		
		Since $\sigma_i$ is measurable by induction hypothesis, all conditions except the last one define events in $\mathcal{F}_{\mathsf{fin}}$. To see that this is also the case for the condition that $p_i(\rho,\tau)\geq p_i(\rho,\tau')$ for any $\tau'\in T$, observe that this event can be rewritten as a projection onto the set of all finite paths of length $i$ of the set of all infinite runs $\tilde{\rho}$ satisfying the following property
		\begin{equation*}
		\begin{split}
		\mathbb{E}^{\sigma_i}\Big[ \mathbb{I}_{\Term=\infty \land \tau^{\rho}_i=\tau} &| \sigma(\mathcal{F}_i,\mathbb{I}_{\tau^{\rho}_i=\tau}) \Big](\tilde{\rho}) \\
		&\geq \max_{\tau'\in T} \mathbb{E}^{\sigma_i}\Big[ \mathbb{I}_{\Term=\infty \land \tau^{\rho}_i=\tau'} | \sigma(\mathcal{F}_i,\mathbb{I}_{\tau'^{\rho}_i=\tau'}) \Big](\tilde{\rho}),
		\end{split}
		\end{equation*}
		where $\sigma(\mathcal{F}_i,\mathbb{I}_{\tau'^{\rho}_i=\tau'})$ is the smallest $\sigma$-algebra containing $\mathcal{F}_i$ and the $\sigma$-algebra generated by the random variable $\mathbb{I}_{\tau'^{\rho}_i=\tau'}$. All random variables involved in the above inequality are $\mathcal{F}$-measurable by measurability of $\sigma_i$, and so the set of infinite runs satisfying this conditions is in $\mathcal{F}$. This shows that the last condition also defines a measurable event. Thus, $A_{\tau,=i,l}\in\mathcal{F}_{\mathsf{fin}}$ and the claim follows.
		
		\item We need to show that $\mathbb{P}^{\sigma_{i+1}}[\Term=\infty] \geq \mathbb{P}^{\sigma_i}[\Term=\infty]$ for each $i$. To do this, it suffices to show that $\mathbb{P}^{\sigma_{i+1}}[\Term=\infty \land \loc^{\rho}_i=\loc] \geq \mathbb{P}^{\sigma_i}[\Term=\infty \land \loc^{\rho}_i=\loc]$ for each location $\loc\in L$ where the event $\{\loc^{\rho}_i=\loc\}$ denotes that $\loc$ is the $(i+1)$-st location along the program run. If we show this, by taking the sum over all locations on both sides of the inequality, the claim follows.
		
		Fix a location $\loc$. Then
		\begin{equation}\label{eq:bigsequence}
		\begin{split}
		\mathbb{P}^{\sigma_{i+1}}[\Term=\infty \land \loc^{\rho}_i=\loc] &= \mathbb{E}^{\sigma_{i+1}}[\mathbb{I}_{\Term=\infty}\cdot\mathbb{I}_{\loc^{\rho}_i=\loc}]  \\
		&= \mathbb{E}^{\sigma_{i+1}}[\mathbb{E}^{\sigma_{i+1}}[\mathbb{I}_{\Term=\infty}\cdot\mathbb{I}_{\loc^{\rho}_i=\loc} | \mathcal{F}_i]] \\
		&= \mathbb{E}^{\sigma_{i+1}}[\mathbb{I}_{\loc^{\rho}_i=\loc} \cdot \mathbb{E}^{\sigma_{i+1}}[\mathbb{I}_{\Term=\infty} | \mathcal{F}_i]] \\
		&= \mathbb{E}^{\sigma_i}[\mathbb{I}_{\loc^{\rho}_i=\loc} \cdot \mathbb{E}^{\sigma_{i+1}}[\mathbb{I}_{\Term=\infty} | \mathcal{F}_i]].
		\end{split}
		\end{equation}
		The first equality holds since the probability of the event is equal to the expected value of its indicator function. The second equality holds by the definition of conditional expectation. The third equality holds since we may take out from the conditional expectation any variable that is measurable w.r.t.~the $\sigma$-algebra that we condition on~\cite{Williams:book}. The fourth inequality holds since the probability measures defined by $\sigma_{i+1}$ and $\sigma_i$ by construction agree on $\mathcal{F}_i$-measurable sets.
		
		But, by construction we have $\mathbb{E}^{\sigma_{i+1}}[\mathbb{I}_{\Term=\infty} | \mathcal{F}_i](\rho)\geq \mathbb{E}^{\sigma_{i}}[\mathbb{I}_{\Term=\infty} | \mathcal{F}_i](\rho)$ for each infinite run $\rho$, because in the $(i+1)$-st configuration $\sigma_{i+1}$ picks the transition which maximizes the probability of non-termination among all transitions in $\support(\sigma_i(\rho_i))$. Hence, plugging this back into eq.~\eqref{eq:bigsequence} we conclude
		\begin{equation*}
		\begin{split}
		\mathbb{P}^{\sigma_{i+1}}[\Term=\infty \land \loc^{\rho}_i=\loc] &= \mathbb{E}^{\sigma_{i+1}}[\mathbb{I}_{\Term=\infty}\cdot\mathbb{I}_{\loc^{\rho}_i=\loc}]  \\
		&= \mathbb{E}^{\sigma_i}[\mathbb{I}_{\loc^{\rho}_i=\loc} \cdot \mathbb{E}^{\sigma_{i+1}}[\mathbb{I}_{\Term=\infty} | \mathcal{F}_i]] \\
		&\geq \mathbb{E}^{\sigma_i}[\mathbb{I}_{\loc^{\rho}_i=\loc} \cdot \mathbb{E}^{\sigma_i}[\mathbb{I}_{\Term=\infty} | \mathcal{F}_i]] \\
		&= \mathbb{P}^{\sigma_i}[\Term=\infty \land \loc^{\rho}_i=\loc],
		\end{split}
		\end{equation*}
		where the last equality follows by the same sequence of equalities as in eq.~\eqref{eq:bigsequence} with $i+1$ replaced by $i$. This proves the claim.
	\end{enumerate}
	
	Hence, schedulers $\sigma=\sigma_0,\sigma_1,\sigma_2,\dots$ satisfy all the desired properties. Finally, note that by construction $\sigma^{\ast}$ is deterministic too. We are left to show that $\sigma^{\ast}$ is measurable and that $\mathbb{P}^{\sigma^{\ast}}_{\locinit,\vecinit}[\Term = \infty] > 0$.
	
	To see that $\sigma^{\ast}$ is measurable, we need to show that $(\sigma^{\ast})^{-1}(\tau)$ is in the $\sigma$-algebra of finite runs for each transition $\tau$. This follows since, if we write $(\sigma^{\ast})^{-1}(\tau) = \cup_{i\in\mathbb{N}_0}A_i$ where $A_i$ is the set of all finite runs in $(\sigma^{\ast})^{-1}(\tau)$ of length $i$, we have that each $A_i$ is in the $\sigma$-algebra since $\sigma^{\ast}$ coincides with $\sigma_i$ on histories of length at most $i$.
	
	To see that $\mathbb{P}^{\sigma^{\ast}}[\Term=\infty] >0$, suppose that on the contrary $\mathbb{P}^{\sigma^{\ast}}[\Term = \infty] = 0$ so $\mathbb{P}^{\sigma^{\ast}}[\Term < \infty] = 1$. Let $\delta = \mathbb{P}^{\sigma}[\Term = \infty]>0$, where $\sigma$ is the potentially non-deterministic scheduler from the proposition statement. Since we have $\mathbb{P}^{\sigma^{\ast}}[\Term \leq i] \rightarrow \mathbb{P}^{\sigma^{\ast}}[\Term < \infty]$ as $i\rightarrow\infty$ by the Monotone Convergence Theorem~\cite{Williams:book}, there exists $i\in\mathbb{N}_0$ such that $\mathbb{P}^{\sigma^{\ast}}[\Term \leq i]\geq 1-\delta/2$. But
	\[ \mathbb{P}^{\sigma^{\ast}}[\Term \leq i] \leq \mathbb{P}^{\sigma_i}[\Term < \infty] \leq \mathbb{P}^{\sigma_0}[\Term < \infty] = 1-\delta \]
	as $\mathbb{P}^{\sigma_i}[\Term = \infty] \geq \mathbb{P}^{\sigma_0}[\Term = \infty]$ by the above monotonicity property of $(\mathbb{P}^{\sigma_i}[\Term = \infty])_{i=0}^{\infty}$. This gives contradiction, hence $\mathbb{P}^{\sigma^{\ast}}[\Term=\infty] >0$ as claimed.
\end{proof}

\section{Proof of Lemma~\ref{lemma:technical}}
\label{sec:appboundedtechnical}

Let $\boldsymbol{\lem}$ be as in the lemma statement. To prove the lemma claim, we need to show that there exists $K>0$ such that the LEM $\boldsymbol{\lem}'$ of the same dimension as $\boldsymbol{\lem}$, and defined via
\[ \lem'_j(\loc,\mathbf{x}) = \lem_j(\loc,\mathbf{x})+K \]
for each component $j$ and state $(\loc,\mathbf{x})$, is a LinGLexRSM map in $\pCFG$ supported by $I$ (with the level map being the same as for $\boldsymbol{\lem}$).

Note that increasing $\boldsymbol{\lem}$ pointwise by a constant $K>0$ preserves the $\pnneg$, $\prank$ conditions for each transition in $\pCFG$, as well as $\expneg$ for each transition of probabilistic branching. Hence, we are left to show that there exists $K>0$ such that for any transition $\tau$ which is not a transition of probabilistic branching we have
\[ \expneg(\boldsymbol{\lem}',\tau) \equiv \mathbf{x}\in I(\loc)\cap\guards(\tau) \Rightarrow \text{min-pre}_{\lem'_j,S^{\leq j-1}_{\mathsf{lev}}}^\tau(\loc,\vec{x}) \geq 0 \text{ for all } 1 \leq j \leq \mathsf{lev}(\tau). \]

Since the LinPP that induces $\pCFG$ satisfies the BSP and since all non-deterministic assignments are defined by closed intervals, there exists $N>0$ such that $\mathbb{P}_{X\sim d}[|X|>N]=0$ for each distribution $d\in \mathcal{D}$, and $[a,b]\subseteq[-N,N]$ for each interval $[a,b]$ appearing in non-deterministic assignments.

We claim that $K=2\cdot N\cdot \text{max-coeff}(\boldsymbol{\lem})$ satisfies the claim, where $\text{max-coeff}(\boldsymbol{\lem})$ is the maximal absolute value of a coefficient appearing in any expression $\boldsymbol{\lem}(\loc')$ for any location $\loc'$.

To prove this, let $\tau$ be a transition which is not a transition of probabilistic branching, and let $\loc$ and $\loc_1$ be its source and target location, respectively. Let $\mathbf{x}\in I(\loc)\cap\guards(\tau)$ and let $1\leq j\leq \mathsf{lev}(\tau)$. In order to prove that $\text{min-pre}_{\lem'_j,S^{\leq j-1}_{\mathsf{lev}}}^\tau(\loc,\vec{x}) \geq 0$ holds, we distinguish between three cases:
\begin{enumerate}
    \item $\updates(\tau)=\bot$ or $\updates(\tau)=(i,u)$ where $u$ is a linear expression with no sampling instruction. Then $(\loc,\mathbf{x})$ has a single successor state $(\loc_1,\mathbf{x}_1)$ upon executing $\tau$.
    \begin{itemize}
    \item If the level of $(\loc_1,\mathbf{x}_1)$ is at least $j$, then $S^{\leq j-1}_{\mathsf{lev}}$ contains no successor states and $\text{min-pre}_{\lem_j',S^{\leq j-1}_{\mathsf{lev}}}^\tau(\loc,\vec{x})=0$ as the integration is performed over the empty set.
    
    \item Otherwise, $S^{\leq j-1}_{\mathsf{lev}}$ contains $(\loc_1,\mathbf{x}_1)$ and
    \[ \text{min-pre}_{\lem_j',S^{\leq j-1}_{\mathsf{lev}}}^\tau(\loc,\vec{x})=\lem_j'(\loc_1,\mathbf{x}_1)=\lem_j(\loc_1,\mathbf{x}_1)+K=\text{min-pre}_{\lem_j}^\tau(\loc,\vec{x})+K\geq 0, \]
    where the inequality $\text{min-pre}_{\lem_j}^\tau(\loc,\vec{x}) \geq 0$ holds since $\weakexpneg(\boldsymbol{\lem},\tau)$.
    \end{itemize}
    
    \item If $\updates(\tau)=(i,u)$ where $u$ is a linear expression which contains sampling from a distribution $d\in \mathcal{D}$, we may write $u=u'+X$, where $u'$ is the linear expression part of $u$ with no distribution samplings, and $X\sim d$. Then
\begin{equation*}
\begin{split}
    &\text{min-pre}_{\lem'_j,S^{\leq j-1}_{\mathsf{lev}}}^{\tau}(\loc,\vec{x}) = \mathbb{E}_{X\sim d}\Big[\lem'_j(\loc_1,\vec{x}[i\leftarrow u'+X])\cdot \mathbb{I}( \text{next state has level } \leq j-1) \Big] \\
    &= \mathbb{E}_{X\sim d}\Big[(\lem_j(\loc_1,\vec{x}[i\leftarrow u'])+K+\text{coeff}[i]\cdot X) \cdot \mathbb{I}(\text{next state has level } \leq j-1)\Big]\\
    &= \mathbb{E}_{X\sim d}\Big[(\lem_j(\loc_1,\vec{x}[i\leftarrow u'])+\text{coeff}[i]\cdot\mathbb{E}[X]) \cdot \mathbb{I}(\text{next state has level } \leq j-1)\Big]\\
    &\hspace{0.3cm} + \mathbb{E}_{X\sim d}\Big[(K-\text{coeff}[i]\cdot\mathbb{E}[X]+\text{coeff}[i]\cdot X) \cdot \mathbb{I}(\text{next state has level } \leq j-1)\Big]\\
    &\geq \mathbb{E}_{X\sim d}\Big[\text{min-pre}_{\lem_j}^\tau(\loc,\vec{x}) \cdot \mathbb{I}(\text{next state has level } \leq j-1)\Big]\\
    &\hspace{0.3cm} + \mathbb{E}_{X\sim d}\Big[(K-2\cdot N\cdot \text{max-coeff}(\boldsymbol{\lem})) \cdot \mathbb{I}(\text{next state has level } \leq j-1)\Big]\\
    &\geq 0,
\end{split}
\end{equation*}
where $\text{min-pre}_{\lem_j}^\tau(\loc,\vec{x})\geq 0$ holds since $\weakexpneg(\boldsymbol{\lem},\tau)$, and $\mathbb{E}[X]\geq-N$ holds and $X\geq-N$ holds almost-surely since $\mathbb{P}_{X\sim d}[|X|>N]=0$ by the definition of $N$.

\item If $\updates(\tau)=(i,u)$ where $u$ is an interval $[a,b]$ defining a non-deterministic assignment, we have
\begin{equation*}
\begin{split}
    \text{min-pre}_{\lem'_j,S^{\leq j-1}_{\mathsf{lev}}}^{\tau}(\loc,\vec{x}) &= \inf_{X\in[a,b]\land (\loc_1,\vec{x}(i\leftarrow X))\in S^{\leq j-1}_{\mathsf{lev}}}\Big[\lem'_j(\loc_1,\vec{x}[i\leftarrow X])\Big] \\
    &\geq \inf_{X\in[a,b]}\Big[\lem'_j(\loc_1,\vec{x}[i\leftarrow X])\Big] \\
    &= \text{min-pre}_{\lem_j}^\tau(\loc,\vec{x}) + K \\
    &\geq 0,
\end{split}
\end{equation*}
where $\text{min-pre}_{\lem_j}^\tau(\loc,\vec{x})\geq 0$ holds since $\weakexpneg(\boldsymbol{\lem},\tau)$ and $K \geq 0$ by definition.
\end{enumerate}

\section{Details on our Algorithm in Section~\ref{sec:boundeddis}}
\label{app:algo}

\begin{algorithm}[t]
	\SetKwInOut{Input}{input}\SetKwInOut{Output}{output}
	\DontPrintSemicolon
	
	\Input{A LinPP $\pCFG$ with the BSP, linear invariant $I$.}
	\Output{LinGLexRSM map supported by $I$ if it exists, otherwise ''No LinGLexRSM map''}
	$\mathcal{T}$ $\longleftarrow$ all transitions in $\pCFG$; $d$ $\longleftarrow$ $0$\\
	\While{$\mathcal{T}$ is non-empty}{
		$d$ $\longleftarrow$ $d+1$\\
		construct $\mathcal{LP}_{\mathcal{T}}$\label{aline:reduce}\\
		\If{$\mathcal{LP}_{\mathcal{T}}$ is feasible}{
			$\lem_d$ $\longleftarrow$ LEM defined by the optimal solution of $\mathcal{LP}_{\mathcal{T}}$\\
			$\mathcal{T}$ $\longleftarrow$ $\mathcal{T}\backslash \{\tau\in\mathcal{T}\mid \tau \text{ is 1-ranked by } \lem_d\}$}\label{aline:prune}
		\lElse{\Return No LinGLexRSM map}
	}
	$\max \longleftarrow \text{max-coeff}(\boldsymbol{\lem})$\\
	$N \longleftarrow$ constant such that all distributions and intervals supported in $[-N,N]$\\
	\For{$1\leq j\leq d$}{
	    $\lem_j \longleftarrow \lem_j + 2\cdot N\cdot \max$}
	\Return{$(\lem_1,\dots,\lem_d)$}
	\caption{Synthesis of LinGLexRSM maps in LinPPs with the BSP.}
	\label{algo:glexrsm-map}
\end{algorithm}

The pseudocode of our algorithm is presented in Algorithm~\ref{algo:glexrsm-map} (where $N$ and $\text{max-coeff}(\boldsymbol{\lem})$ are defined analogously as in Appendix~\ref{sec:appboundedtechnical}). 

Next, we show how condition 4 from the main text can be encoded using linear constraints. Let $\tau=(\loc,\delta)\in \mathcal{T}$ be a transition of probabilistic branching. Then $\support(\delta)=(\loc_1,\loc_2)$ and $ \updates(\tau) = \bot $, so for any $\mathbf{x}\in I(\loc)\cap \guards(\tau)$ we have that
    \[ \text{pre}_{\lem,S}^\tau(\loc,\vec{x}) = \delta(\loc_1)\cdot\lem(\loc_1,\mathbf{x})\cdot \mathbb{I}(S)(\loc_1,\mathbf{x})+\delta(\loc_2)\cdot\lem(\loc_2,\mathbf{x})\cdot \mathbb{I}(S)(\loc_2,\mathbf{x}), \]
    i.e.~we include the term $\delta(\loc_i)\cdot\lem(\loc_i,\mathbf{x})$ for $i\in\{1,2\}$ whenever $(\loc_i,\mathbf{x})\in S$. Hence, to encode condition 4 for $\tau$, define $\guards_1=\neg(\lor_{\tau=(\loc_1,)\in\mathcal{T}}\guards(\tau))$ and $\guards_2=\neg(\lor_{\tau=(\loc_2,)\in\mathcal{T}}\guards(\tau))$, and encode the following $3$ conditions:
    \begin{itemize}
        \item $\forall \mathbf{x}.\, \mathbf{x}\in I(\loc)\cap \guards(\tau)\cap \guards_1\cap \guards_2 \Rightarrow\delta(\loc_1)\cdot\lem(\loc_1,\mathbf{x})+\delta(\loc_2)\cdot\lem(\loc_2,\mathbf{x}) \geq 0,$
        \item $\forall \mathbf{x}.\, \mathbf{x}\in I(\loc)\cap \guards(\tau)\cap \guards_1\cap \neg\guards_2 \Rightarrow\delta(\loc_1)\cdot\lem(\loc_1,\mathbf{x}) \geq 0,$ and
        \item $\forall \mathbf{x}.\, \mathbf{x}\in I(\loc)\cap \guards(\tau)\cap \neg \guards_1\cap \guards_2 \Rightarrow\delta(\loc_2)\cdot\lem(\loc_2,\mathbf{x}) \geq 0.$
    \end{itemize}
    Each condition can be encoded via linear constraint as in~\cite{ADFG10:lexicographic,AgrawalC018}. Clearly, the size of the encoding is polynomial.
    
Note that the negations in $\guards_1$ and $\guards_2$ might result in strict inequalities appearing in the above constraints. However, it was shown in~\cite{CFNH16:prob-termination} that this is not an issue for the Farkas' lemma (FL) conversion. Indeed, Lemma~1 in~\cite{CFNH16:prob-termination} shows that, whenever a system of linear inequalities on the LHS of a constraint is feasible, the strict inequalities may without loss of generality be replaced by non-strict inequalities. On the other hand, Lemma~2 in~\cite{CFNH16:prob-termination} shows that this feasibility check can be done in polynomial time.
    
The following theorem establishes soundness and completeness of our algorithm in Section~\ref{sec:boundeddis}.

\begin{theorem}[Soundness and completeness]\label{thm:algorithm}
	If Algorithm~\ref{algo:glexrsm-map} computes an LEM $\boldsymbol{\lem}$, then $\boldsymbol{\lem}$ is a LinGLexRSM map supported by $I$ and $\pCFG$ is a.s.~terminating. Moreover, whenever $\pCFG$ admits a LinGLexRSM map supported by $I$, Algorithm~\ref{algo:glexrsm-map} finds such a map of minimal dimension and proves a.s.~termination of $\pCFG$.
\end{theorem}
\begin{proof}
We first prove that the algorithm is sound, i.e.~that $\boldsymbol{\lem}$ is a LinGLexRSM map supported by $I$, and thus that $\pCFG$ is a.s.~terminating. Let $k$ be the total number of algorithm iterations, so that $\boldsymbol{\lem}=(\lem_1,\dots,\lem_k)$. Define the level map $\mathsf{lev}:\gentransitions\rightarrow \{0,1\dots,k\}$ with the self loop at $\locterm$ having level $0$, and for any other transition $\tau$ we define $\mathsf{lev}(\tau)$ as the index of algorithm iteration in which it was removed from $\mathcal{T}$. The fact that $\boldsymbol{\lem}$ computed in lines 1-8 in Algorithm~\ref{algo:glexrsm-map} satisfies $\pnneg$, $\prank$, $\expneg$ for transitions of probabilistic branching and $\weakexpneg$ for all other transitions then easily follows from conditions imposed by the algorithm in each iteration. From the proof of Lemma~\ref{lemma:technical}, it then follows that $\boldsymbol{\lem}$ obtained upon increasing each component by a constant term in lines 9-13 satisfies $\expneg$ for every transition. Hence $\boldsymbol{\lem}$ is a LinGLexRSM map supported by $I$ and this concludes the soundness proof.

To prove completeness as well as the minimality of dimension, we observe that a pointwise sum of two LinGLexRSM maps supported by $I$ is also a LinGLexRSM map supported by $I$. This follows by linearity of integration and therefore the pre-expectation operator. The argument is straightforward, thus we omit it. However, this simple observation will be central in the rest of the proof.

Suppose first that the program admits a LinGLexRSM $\boldsymbol{\lem}'=(\lem_1',\dots,\lem_m')$ supported by $I$. We show that Algorithm~\ref{algo:glexrsm-map} then finds one such LinGLexRSM map (up to a constant term), hence the algorithm is complete. We prove this by contradiction. Suppose that the algorithm stops after the $k$-th iteration, after having computed $(\lem_1,\dots,\lem_k)$ but with $\mathcal{T}$ still containing at least one transition. Then $(\lem_1,\dots,\lem_k)$ does not rank every transition in the pCFG. Thus, $\boldsymbol{\lem}'$ ranks strictly more transitions than $(\lem_1,\dots,\lem_k)$. We distinguish two cases:
\begin{enumerate}
\item There exists the smallest $1\leq j\leq \min\{k,m\}$ such that
\begin{itemize}
	\item for each $1\leq j'<j$, $\lem_{j'}$ and $\lem_{j'}'$ would rank exactly the same set of transitions if computed by the algorithm in the $j'$-th iteration, but
	\item $\lem_{j}'$ ranks a transition which is not ranked by $\lem_j$ in the $j$-th iteration of the algorithm.
\end{itemize}
Then the algorithm could have ranked strictly more transitions by computing $\lem_j+\lem_j'$ instead of $\lem_j$, which contradicts the maximality condition for computing new components that is imposed by the algorithm.
\item There is no such index. But then, since $\boldsymbol{\lem}'=(\lem_1',\dots,\lem_m')$ is the LinGLexRSM supported by $I$, it must follow that $m>k$ and that $\lem'_{k+1}$ would satisfy all the conditions imposed by the algorithm in the $(k+1)$-st iteration and it would rank at least $1$ new transition, thus the algorithm couldn't terminate after iteration $k$.
\end{enumerate}
Thus, in both cases we reach contradiction, and the completeness claim on Algorithm~\ref{algo:glexrsm-map} holds.

\smallskip Minimality of dimension is proved analogously as completeness, by contradiction. If there exists a LinGLexRSM map w.r.t.~$S$ supported by $I$ of dimension strictly smaller by that found by the algorithm, we can use it analogously as above to show that at some iteration the algorithm could have ranked a strictly larger number of transitions, contradicting the maximality condition for computing new components that is imposed by the algorithm. Thus the minimality of dimension claim follows.
\end{proof}

\section{Proof of Lemma~\ref{lemma:technical2}}\label{app:technical2}

Let $\boldsymbol{\lem}$ be an LEM whose existence is assumed in the lemma statement. Analogously as in the proof of Lemma~\ref{lemma:technical}, we may increase $\boldsymbol{\lem}$ by a constant term in order to ensure that all transitions satisfy $\expneg$, except for maybe those that in the variable update involve sampling from a distribution of unbounded support. So without loss of generality assume that $\boldsymbol{\lem}$ satisfies $\expneg$ for all other transitions. Denote the set of all transitions in $\pCFG$ that involve sampling from distributions of unbounded support by $\transitions^{\text{unb}}$.

As before, denote by $\text{max-coeff}(\boldsymbol{\lem})$ the maximal absolute value of a coefficient appearing in $\boldsymbol{\lem}$. Also, define $N$ analogously as in the proof of Lemma~\ref{lemma:technical}, i.e. for all distributions of bounded support that appear in sampling instructions and for all bounded intervals appearing in non-deterministic assignments, we have that they are supported in $[-N,N]$. Finally, since we assume that each distribution appearing in sampling instructions is integrable, for each $d\in\mathcal{D}$ we have $\mathbb{E}_{X\sim d}[|X|]<\infty$. Thus, by triangle inequality we also have $\mathbb{E}_{X\sim d}[|X-\mathbb{E}[X]|]<\infty$. Hence, as $\mathbb{E}_{X\sim d}[|X-\mathbb{E}[X]|\cdot \mathbb{I}(|X-\mathbb{E}[X]|< k)]\rightarrow \mathbb{E}_{X\sim d}[|X-\mathbb{E}[X]|]$ as $k\rightarrow\infty$ by the Monotone Convergence Theorem~\cite{Williams:book}, for each $d\in\mathcal{D}$ there exists $k(d)\in\mathbb{N}$ such that
\[ \mathbb{E}_{X\sim d}\Big[|X-\mathbb{E}[X]|\cdot \mathbb{I}\Big(|X-\mathbb{E}[X]|\geq k\Big)\Big]<\frac{1}{2\cdot\text{max-coeff}(\boldsymbol{\lem})}. \]
for all $k\geq k(d)$. Define $K=\max_{d\in\mathcal{D}}k(d)$, which is finite as $\mathcal{D}$ is finite.

Next, define the set $U\subseteq \{1,2,\dots,\dim(\boldsymbol{\lem})\}\times \locs$ of pairs of indices of components of $\boldsymbol{\lem}$ and locations in $\pCFG$ as follows:
\[ U = \{(j,\loc')\mid \exists \tau\in\transitions^{\text{unb}} \text{ s.t. } \mathsf{lev}(\tau)=j \text{ and } \loc' \text{ is the target location of } \tau \}. \]
Thus, $U$ is the set of pairs of indices of components of $\boldsymbol{\lem}$ and locations in $\pCFG$ on which the condition $\unbounded$ imposes additional template restrictions.

We now define an LEM $\boldsymbol{\lem}'$ which is of the same dimension as $\boldsymbol{\lem}$, and for each component $\lem_j'$ we define
\begin{equation*}
    \lem_j'(\loc,\mathbf{x}) = \begin{cases}
        0 &\text{if } (j,\loc)\in U \text{ and } \lem_j(\loc,\mathbf{x})< -C,\\
        2\cdot\lem_j(\loc,\mathbf{x})+2\cdot C &\text{otherwise},
    \end{cases}
\end{equation*}
where $C>0$ is a constant to be determined. We claim that there exists $C>0$ for which $\boldsymbol{\lem}'$ is a piececwise linear GLexRSM map supported by $I$ (with the level map being the same as for $\boldsymbol{\lem}$). The fact that $\boldsymbol{\lem}'$ is piecewise linear for every $C>0$ is clear from its definition, so we are left to verify that there exists $C>0$ for which $\pnneg(\boldsymbol{\lem}',\tau)$, $\prank(\boldsymbol{\lem}',\tau)$ and $\expneg(\boldsymbol{\lem}',\tau)$ hold for each transition $\tau$.

\medskip\noindent{\em Transitions not in $\transitions^{\text{unb}}$.} We show that, for transitions not in $\transitions^{\text{unb}}$, the claim holds for every $C>2\cdot N\cdot \text{max-coeff}(\boldsymbol{\lem})$.

First, suppose that $\tau$ is a transition of probabilistic branching. Let $\loc$ be its source location and $\loc_1$, $\loc_2$ its target locations. Since we assume that $\pCFG$ is induced by a program in LinPP$^\ast$ meaning that neither $\loc_1$ nor $\loc_2$ are target locations of any transition in $\transitions^{\text{unb}}$, by definition of $U$ and $\boldsymbol{\lem}'$ we must have $\lem_j'(\loc_1,\mathbf{x})=2\cdot\lem_j(\loc_1,\mathbf{x})+2\cdot C$ and $\lem_j'(\loc_2,\mathbf{x})=2\cdot\lem_j(\loc_2,\mathbf{x})+2\cdot C$ for each component $\lem_j'$ and each $\mathbf{x}$. On the other hand, the piecewise linear transformation defining $\boldsymbol{\lem}'$ ensures that $\lem_j'(\loc,\mathbf{x})\geq 2\cdot\lem_j(\loc,\mathbf{x})+2\cdot C$ for each component $\lem_j'$ and each $\mathbf{x}$. Hence, it is easy to see that $\pnneg(\boldsymbol{\lem}',\tau)$, $\prank(\boldsymbol{\lem}',\tau)$ and $\expneg(\boldsymbol{\lem}',\tau)$ all hold as they hold for $\boldsymbol{\lem}$. Note that this proof allows any $C>0$.

Next, let $\tau\not\in\transitions^{\text{unb}}$ be a transition which is also not a transition of probabilistic branching. Denote by $\loc$ its source location, $\loc_1$ its target location, $\mathbf{x}\in I(\loc)\cap \guards(\tau)$, and $(\loc_1,\mathbf{x}_1)$ some state that is reachable from $(\loc,\mathbf{x})$ by executing $\tau$. We claim that, for each $1\leq j\leq \mathsf{lev}(\tau)$,
\[ \lem_j(\loc_1,\mathbf{x}_1)\geq -2\cdot N\cdot \text{max-coeff}(\boldsymbol{\lem}), \] 
By definition of $\boldsymbol{\lem}'$, if $C>2\cdot N\cdot \text{max-coeff}(\boldsymbol{\lem})$ this would imply that $\lem_j'(\loc_1,\mathbf{x}_1)=2\cdot\lem_j(\loc_1,\mathbf{x}_1)+2\cdot C$ for each successor state $(\loc_1,\mathbf{x}_1)$. Since $\lem_j'(\loc,\mathbf{x})\geq 2\cdot\lem_j(\loc,\mathbf{x})+2\cdot C$, it is again easy to see that $\pnneg(\boldsymbol{\lem}',\tau)$, $\prank(\boldsymbol{\lem}',\tau)$ and $\expneg(\boldsymbol{\lem}',\tau)$ all remain true as they are true for $\boldsymbol{\lem}$.

To prove the claim, fix $1\leq j\leq \mathsf{lev}(\tau)$ and a successor state $(\loc_1,\mathbf{x}_1)$. By the condition $\weakexpneg(\boldsymbol{\lem},\tau)$, we have $\text{min-pre}^{\tau}_{\lem_j}(\loc,\mathbf{x})\geq 0$.
If the update element of $\tau$ does not contain a sampling instruction, then we must have $\lem_j(\loc_1,\mathbf{x}_1)\geq\text{min-pre}^{\tau}_{\lem_j}(\loc,\mathbf{x})$ by definition of $\text{min-pre}$, and the claim follows. Otherwise, suppose that $\updates(\tau)=(i,u)$ with $u=u'+X$ where $u'$ is a linear expression without sampling instructions and $X\sim d$ where $d$ is a distribution of bounded support (recall, we assumed that $\tau\not\in\transitions^{\text{unb}}$). Then, by linearity of $\lem_j$ we easily see that
\begin{equation*}
\begin{split}
    \lem_j(\loc_1,\mathbf{x}_1) &= \text{min-pre}^{\tau}_{\lem_j}(\loc,\mathbf{x})+\text{coeff}[i]\cdot X-\text{coeff}[i]\cdot \mathbb{E}[X]\\
    &\geq  \text{coeff}[i]\cdot X-\text{coeff}[i]\cdot \mathbb{E}[X]\\
    &\geq -2\cdot |\text{coeff}[i]| \cdot N\\
    &\geq -2\cdot \text{max-coeff}(\boldsymbol{\lem}) \cdot N,
\end{split}    
\end{equation*}
as claimed. The first inequality follows from $\weakexpneg(\boldsymbol{\lem},\tau)$, and the rest follows by definition of $N$ and the assumption that $d$ has bounded support. Here, we used $\text{coeff}[i]$ to denote the coefficient in $\lem_j$ of the variable with index $i$ at location $\loc_1$.

\medskip\noindent{\em Transitions in $\transitions^{\text{unb}}$.} Let $\tau\in \transitions^{\text{unb}}$, let $\loc$ be its source location and $\loc_1$ its target location. We claim that each of the conditions $\pnneg(\boldsymbol{\lem}',\tau)$, $\prank(\boldsymbol{\lem}',\tau)$ and $\expneg(\boldsymbol{\lem}',\tau)$ holds if $C>K\cdot \text{max-coeff}(\boldsymbol{\lem})$.

\begin{enumerate}
    \item $\pnneg(\boldsymbol{\lem}',\tau)$: By definition of $\boldsymbol{\lem}'$, for each component $\lem_j'$ and $\mathbf{x}$ we have $\lem_j'(\loc,\mathbf{x})\geq 2\cdot\lem_j(\loc,\mathbf{x})+C$. Hence, the claim follows since $\pnneg(\boldsymbol{\lem},\tau)$ holds.
    
    \item $\prank(\boldsymbol{\lem}',\tau)$: If $\mathbf{x}\in I(\loc)\cap \guards(\tau)$, we need to show that for $1\leq j< \mathsf{lev}(\tau)$ we have
    $\lem'_j(\loc,\mathbf{x}) \geq \text{max-pre}^{\tau}_{\lem'_j}(\loc,\mathbf{x})$,
    and that
    $\lem'_{\mathsf{lev}(\tau)}(\loc,\mathbf{x}) \geq \text{max-pre}^{\tau}_{\lem'_{\mathsf{lev}(\tau)}}(\loc,\mathbf{x})+1$.
    
    First, fix $1\leq j<\mathsf{lev}(\tau)$. From the $\unbounded$ condition, we know that the coefficient in $\lem_j$ of the variable updated by $\tau$ at $\loc_1$ is $0$. Hence, $\lem_j$ has the same value at each successor state of $(\loc,\mathbf{x})$ upon executing $\tau$ which is thus equal to $\text{max-pre}^{\tau}_{\lem_j}(\loc,\mathbf{x})$. By the $\weakexpneg(\boldsymbol{\lem},\tau)$ this value has to be nonnegative. Hence, the value of $\lem_j'$ is also the same at each successor state and equal to $2\cdot\text{max-pre}^{\tau}_{\lem_j}(\loc,\mathbf{x})+C$. Thus, as $\lem_j'(\loc,\mathbf{x})\geq 2\cdot\lem_j(\loc,\mathbf{x})+C$, the desired inequality holds as $\prank(\boldsymbol{\lem},\tau)$ holds.
    
    We now prove that $\lem'_{\mathsf{lev}(\tau)}(\loc,\mathbf{x}) \geq \text{max-pre}^{\tau}_{\lem'_{\mathsf{lev}(\tau)}}(\loc,\mathbf{x})+1$. Let $\updates(\tau)=(i,u)$ with $u=u'+X$ where $u'$ is a linear expression with no sampling instructions and $X\sim d$ where $d$ has unbounded support. Since $(\mathsf{lev}(\tau),\loc_1)\in U$, we have
    \begin{equation*}
    \begin{split}
        &\text{max-pre}^{\tau}_{\lem'_{\mathsf{lev}(\tau)}}(\loc,\mathbf{x}) = \mathbb{E}_{X\sim d}\Big[\lem'_{\mathsf{lev}(\tau)}(\loc_1,\vec{x}[i\leftarrow u'+X])\Big] \\
        &= \mathbb{E}_{X\sim d}\Big[(2\cdot\lem_{\mathsf{lev}(\tau)}(\loc_1,\vec{x}[i\leftarrow u'+X])+2\cdot C)\cdot\mathbb{I}\Big(\lem_{\mathsf{lev}(\tau)}(\loc_1,\vec{x}[i\leftarrow u'+X])\geq -C\Big) \Big]\\
        &= 2\cdot C+2\cdot \text{max-pre}^{\tau}_{\lem_{\mathsf{lev}(\tau)}}(\loc,\mathbf{x})\\
        &- \mathbb{E}_{X\sim d}\Big[(2\cdot\lem_{\mathsf{lev}(\tau)}(\loc_1,\vec{x}[i\leftarrow u'+X])+2\cdot C)\cdot\mathbb{I}\Big(\lem_{\mathsf{lev}(\tau)}(\loc_1,\vec{x}[i\leftarrow u'+X])< -C\Big) \Big]\\
        &\hspace{1cm} \text{(use that $\text{max-pre}^{\tau}_{\lem_{\mathsf{lev}(\tau)}}(\loc,\mathbf{x})\leq \lem_{\mathsf{lef}(\tau)}(\loc,\mathbf{x})-1$ by $\prank(\boldsymbol{\lem},\tau)$)}\\
        &\leq 2\cdot C+2\cdot \lem_{\mathsf{lef}(\tau)}(\loc,\mathbf{x})-2\\
        &- \mathbb{E}_{X\sim d}\Big[(2\cdot\lem_{\mathsf{lev}(\tau)}(\loc_1,\vec{x}[i\leftarrow u'+X])+2\cdot C)\cdot\mathbb{I}\Big(\lem_{\mathsf{lev}(\tau)}(\loc_1,\vec{x}[i\leftarrow u'+X])< -C\Big) \Big]\\
        &\hspace{1cm} \text{(from definition of $\boldsymbol{\lem}'$, have $2\cdot C+2\cdot \lem_{\mathsf{lef}(\tau)}(\loc,\mathbf{x})\leq \lem'_{\mathsf{lev}(\tau)}(\loc,\mathbf{x})$)}\\
        &\leq \lem'_{\mathsf{lev}(\tau)}(\loc,\mathbf{x})-2\\
        &- \mathbb{E}_{X\sim d}\Big[(2\cdot\lem_{\mathsf{lev}(\tau)}(\loc_1,\vec{x}[i\leftarrow u'+X])+2\cdot C)\cdot\mathbb{I}\Big(\lem_{\mathsf{lev}(\tau)}(\loc_1,\vec{x}[i\leftarrow u'+X])< -C\Big) \Big].
    \end{split}
    \end{equation*}
    Now, $\lem_{\mathsf{lev}(\tau)}(\loc_1,\vec{x}[i\leftarrow u'+X])=\lem_{\mathsf{lev}(\tau)}(\loc_1,\vec{x}[i\leftarrow u'+\mathbb{E}[X])+(X-\mathbb{E}[X])\cdot \text{coeff}[i]=\text{min-pre}^{\tau}_{\lem_{\mathsf{lev}(\tau)}}(\loc,\mathbf{x})+(X-\mathbb{E}[X])\cdot \text{coeff}[i]\geq (X-\mathbb{E}[X])\cdot \text{coeff}[i]$, which holds by linearity of $\boldsymbol{\lem}$ and the last inequality follows from $\weakexpneg(\boldsymbol{\lem},\tau)$. Here, we use $\text{coeff}[i]$ to denote the coefficient at $\loc_1$ of the variable with index $i$ in $\lem_{\mathsf{lev}(\tau)}$. Hence, continuing the above sequence inequalities we have that (note that the integrand is negative on the set over which integration is performed hence, as we have the minus sign outside the integral, the whole expression increases if we further decrease the integrand but enlarge the event over which the integration is performed in a way which keeps the integrand negative)
    \begin{equation*}
        \leq \lem'_{\mathsf{lev}(\tau)}(\loc,\mathbf{x})-2 - \mathbb{E}_{X\sim d}\Big[(2\cdot(X-\mathbb{E}[X])\cdot \text{coeff}[i]+2\cdot C)\cdot\mathbb{I}\Big((X-\mathbb{E}[X])\cdot \text{coeff}[i]< -C\Big) \Big].
    \end{equation*}
    Thus, to conclude that  $\lem'_{\mathsf{lev}(\tau)}(\loc,\mathbf{x}) \geq \text{max-pre}^{\tau}_{\lem'_{\mathsf{lev}(\tau)}}(\loc,\mathbf{x})+1$ it suffices to show
    \[ \mathbb{E}_{X\sim d}\Big[(2\cdot(X-\mathbb{E}[X])\cdot \text{coeff}[i]+2\cdot C)\cdot\mathbb{I}\Big((X-\mathbb{E}[X])\cdot \text{coeff}[i]< -C\Big) \Big]\geq -1. \]
    Now observe that, in order for $(X-\mathbb{E}[X])\cdot \text{coeff}[i]< -C$ to hold we must have $X-\mathbb{E}[X]$ and $\text{coeff}[i]$ be of opposite signs. Therefore, we have
    \begin{equation*}
    \begin{split}
        &\mathbb{E}_{X\sim d}\Big[(2\cdot(X-\mathbb{E}[X])\cdot \text{coeff}[i]+2\cdot C)\cdot\mathbb{I}\Big((X-\mathbb{E}[X])\cdot \text{coeff}[i]< -C\Big) \Big]\\
        &\geq \mathbb{E}_{X\sim d}\Big[(-2|X-\mathbb{E}[X]|\cdot\text{max-coeff}(\boldsymbol{\lem})+2\cdot C)\cdot\mathbb{I}\Big(|X-\mathbb{E}[X]|> C / \text{max-coeff}(\boldsymbol{\lem})\Big) \Big]\\
        &= 2\text{max-coeff}(\boldsymbol{\lem})\cdot\mathbb{E}_{X\sim d}\Big[(-|X-\mathbb{E}[X]|+C / \text{max-coeff}(\boldsymbol{\lem}))\cdot\mathbb{I}\Big(|X-\mathbb{E}[X]|> C / \text{max-coeff}(\boldsymbol{\lem})\Big) \Big]\\
        &\geq 2\text{max-coeff}(\boldsymbol{\lem})\cdot\mathbb{E}_{X\sim d}\Big[(-|X-\mathbb{E}[X]|)\cdot\mathbb{I}\Big(|X-\mathbb{E}[X]|> C / \text{max-coeff}(\boldsymbol{\lem})\Big) \Big]\\
        &\geq -1,
    \end{split}
    \end{equation*}
    where the last inequality holds since $C>K\cdot \text{max-coeff}(\boldsymbol{\lem})$ and by definition of $K$.
    
    \item $\expneg(\boldsymbol{\lem'},\tau)$: Let $\mathbf{x}\in I(\loc)\cap \guards(\tau)$, we show that $\text{min-pre}_{\lem'_j,S^{\leq j-1}_{\mathsf{lev}}}^\tau(\loc,\vec{x}) \geq 0$ for all $ 1 \leq j \leq \mathsf{lev}(\tau)$. For $1\leq j< \mathsf{lev}(\tau)$, by the $\unbounded$ condition we know that, at $\loc_1$, the coefficient in $\lem_j$ of the variable which is updated by $\tau$ is $0$. Hence, the value of $\lem_j$ at all successor states of $(\loc,\mathbf{x})$ upon executing $\tau$ is the same, and is equal to $\text{min-pre}_{\lem_j}^\tau(\loc,\vec{x})$ which is nonnegative by $\weakexpneg(\boldsymbol{\lem},\tau)$.  Therefore, we must have $\lem_j'(\loc_1,\mathbf{x}_1)=2\cdot\lem_j(\loc_1,\mathbf{x}_1)+2\cdot C$ at each state $(\loc_1,\mathbf{x}_1)$ reachable from $(\loc,\mathbf{x})$ by executing $\tau$. Therefore, we also have $\text{min-pre}_{\lem'_j,S^{\leq j-1}_{\mathsf{lev}}}^\tau(\loc,\vec{x}) = 2\cdot \text{min-pre}_{\lem_j,S^{\leq j-1}_{\mathsf{lev}}}^\tau(\loc,\vec{x})+2\cdot C\geq 0$, where the last inequality holds since $\expneg(\boldsymbol{\lem},\tau)$. \\
    For the component $\mathsf{lev}(\tau)$, note that $(\mathsf{lev}(\tau),\loc_1)\in U$. Thus, from our definition of $\lem'$ it follows that $\lem'_{\mathsf{lev}(\tau)}(\loc_1,\mathbf{x}_1)\geq 0$ for every variable valuation $\mathbf{x}_1$. Hence,
    \[ \text{min-pre}_{\lem'_{\mathsf{lev}(\tau)},S^{\leq \mathsf{lev}(\tau)-1}_{\mathsf{lev}}}^\tau(\loc,\vec{x}) \geq 0 \] 
    since it is just an integral of a non-negative function over the set $S^{\leq \mathsf{lev}(\tau)-1}_{\mathsf{lev}}$.
\end{enumerate}

\medskip\noindent{\em Choice of $C$.} From the analysis of all cases above, we see that
\[ C=(2N+K)\cdot\text{max-coeff}(\boldsymbol{\lem})+1 \]
ensures that $\boldsymbol{\lem}'$ is a piecewise linear GLexRSM map, which proves the lemma claim.

\section{Proof of Theorem~\ref{thm:soundness2}}\label{app:soundess2}

We first prove that the algorithm is sound, i.e.~that the LEM $\boldsymbol{\lem}$ which algorithm outputs must satisfy all conditions of Lemma~\ref{lemma:technical2}. Let $k$ be the total number of algorithm iterations, so that $\boldsymbol{\lem}=(\lem_1,\dots,\lem_k)$. Define the level map $\mathsf{lev}:\gentransitions\rightarrow \{0,1\dots,k\}$ with the self loop at $\locterm$ having level $0$, and for any other transition $\tau$ we define $\mathsf{lev}(\tau)$ as the index of algorithm iteration in which it was removed from $\mathcal{T}$. The fact that $\boldsymbol{\lem}$ satisfies $\pnneg$, $\prank$, $\expneg$ for transitions of probabilistic branching and $\weakexpneg$ for all other transitions then easily follows from conditions imposed by the algorithm in each iteration. Furthermore, the way we constructed linear programs $\mathcal{LP}_{\mathcal{T}}^{\text{unb}}$ and $\mathcal{LP}_{\mathcal{T}}^{\tau,\text{unb}}$ for each $\tau\in\transitions^{\text{unb}}\cap\mathcal{T}$ ensures that $\boldsymbol{\lem}$ satisfies $\unbounded$. Hence $\boldsymbol{\lem}$ is an LEM supported by $I$ which satisfies all conditions of Lemma~\ref{lemma:technical2}.

We now prove completeness, i.e.~that for any $(\pCFG,I)$ with $\pCFG$ coming from a program in LinPP$^{\ast}$, Algorithm~\ref{algo:generalprogs} decides the existence of an LEM supported by $I$ which satisfies conditions of Lemma~\ref{lemma:technical2}. First, observe that for any two LEMs supported by $I$ and which satisfy all the conditions in Lemma~\ref{lemma:technical2}, their pointwise sum also satisfies all the conditions in Lemma~\ref{lemma:technical2}. Hence, an argument analogous to that in the proof of Theorem~\ref{thm:algorithm} shows that the algorithm finds an LEM satisfying all the conditions of Lemma~\ref{lemma:technical2} whenever one such LEM exists by observing that whenever an LEM exists but $\mathcal{T}$ is non-empty, either $\mathcal{LP}_{\mathcal{T}}^{\text{unb}}$ or $\mathcal{LP}_{\mathcal{T}}^{\tau,\text{unb}}$ for at least one $\tau\in \transitions^{\text{unb}}\cap\mathcal{T}$ has a solution which $1$-ranks at least one new transition.

Note that due to a fixed ordering of transitions in $\transitions^{\text{unb}}\cap \mathcal{T}$ through which the algorithm iterates, the dimension of the computed LEM which satisfies all the conditions of Lemma~\ref{lemma:technical2} need not be minimal. However, this was not the claim of our theorem.

\end{document}